\documentclass[11pt,letter]{article}

\usepackage{booktabs}
\usepackage{amssymb}
\usepackage{amsmath}
\usepackage{graphicx}
\usepackage{color}
\usepackage{xspace}
\usepackage{enumerate}
\usepackage{hyperref}
\usepackage{numprint}
\usepackage{subcaption}
\usepackage{changepage}
\usepackage{enumerate}
\usepackage{amsthm}
\usepackage{graphicx}
\usepackage{mathrsfs}
\usepackage{bm}
\usepackage[ruled,linesnumbered]{algorithm2e}

\newtheorem{theorem}{\text{Theorem}}
\newtheorem{lemma}{\text{Lemma}}

\newtheorem{corollary}{\text{Corollary}}

\newtheorem{property}{\text{Property}}

\newcommand{\real} {\mathbb{R}}
\newcommand{\eps} {\varepsilon}
\newcommand{\cancel}[1] {}
\newcommand{\norm}[1] {\|#1\|}

\begin{document}
	
\begin{titlepage}
		
\title{Fr\'echet Distance in Subquadratic Time\thanks{Research supported by Research Grants Council, Hong Kong, China (project no.~16208923).}}
		
\author{Siu-Wing Cheng\footnote{Department of Computer Science and Engineering, HKUST, Hong Kong.  Email:~{\tt scheng@cse.ust.hk, haoqiang.huang@connect.ust.hk}} \quad\quad Haoqiang Huang\footnotemark[2]}
		
\date{}
		
\maketitle
		
\begin{abstract}
Let $m$ and $n$ be the numbers of vertices of two polygonal curves in $\real^d$ for any fixed $d$ such that $m \leq n$.  Since it was known in 1995 how to compute the Fr\'{e}chet distance of these two curves in $O(mn\log (mn))$ time, it has been an open problem whether the running time can be reduced to $o(n^2)$ when $m = \Omega(n)$.  In the mean time,  several well-known quadratic time barriers in computational geometry have been overcome: 3SUM, some 3SUM-hard problems, and the computation of some distances between two polygonal curves, including the discrete Fr\'{e}chet distance, the dynamic time warping distance, and the geometric edit distance.  It is curious that the quadratic time barrier for Fr\'{e}chet distance still stands.  We present an algorithm to compute the Fr\'echet distance in $O(mn(\log\log n)^{2+\mu}\log n/\log^{1+\mu} m)$ expected time for some constant $\mu \in (0,1)$.  It is the first algorithm that returns the Fr\'{e}chet distance in $o(mn)$ time when $m = \Omega(n^{\eps})$ for any fixed $\eps \in (0,1]$.
\end{abstract}
		
		\thispagestyle{empty}
	\end{titlepage}

\section{Introduction}

Measuring the similarity between two curves is a fundamental task in many applications such as moving object analysis~\cite{buchin2011detecting} and road network reconstruction~\cite{buchin2017clustering}.
Fr\'{e}chet distance is a popular distance metric for curves as it captures their similarity well~\cite{Alt2009}.

Let $\tau =(v_1, v_2,...,v_n)$ be a polygonal curve in $\mathbb{R}^d$.  A \emph{parameterization} $\rho$ of $\tau$ maps every $t \in [0,1]$ to a point $\rho(t)$ on $\tau$ such that $\rho(0)=v_1$, $\rho(1)=v_n$, and $\rho(t)$ moves monotonically from $v_1$ to $v_n$ as $t$ increases from 0 to 1.  Let $\varrho$ be a parameterization of another polygonal curve $\sigma=(w_1, w_2,..., w_m)$.  The pair $(\rho,\varrho)$ define a \emph{matching} $\mathcal{M}$ that maps $\rho(t)$ with $\varrho(t)$ for all $t \in [0,1]$.  A point in $\tau$ may have multiple matching partners in $\sigma$ and vice versa.  Let $d(x, y)$ be the Euclidean distance between $x$ and $y$ in $\mathbb{R}^d$. The distance between $\tau$ and $\sigma$ realized by $\mathcal{M}$ is $d_{\mathcal{M}}(\tau, \sigma)=\max_{t\in[0,1]}d(\rho(t), \varrho(t))$. The \emph{Fr\'echet distance} $d_F(\sigma,\tau)$ between $\sigma$ and $\tau$ is $\inf_{\mathcal{M}}d_{\mathcal{M}}(\tau, \sigma)$. We call a matching that realizes the Fr\'echet distance a \emph{Fr\'echet matching}.  If we restrict each parameterization to match a vertex of $\tau$ to at least one vertex of $\sigma$, and vice versa, the resulting $\inf_{\mathcal{M}}d_{\mathcal{M}}(\tau, \sigma)$ is known as the \emph{discrete Fr\'{e}chet distance}, which we denote by $\tilde{d}_F(\sigma,\tau)$.  It follows from the definition that $d_F(\sigma,\tau) \leq \tilde{d}_F(\sigma,\tau)$.

We assume throughout this paper that $m \leq n$. Godau~\cite{Godau1991ANM} described the first polynomial-time algorithm for computing $d_F(\sigma,\tau)$ that runs in $O\bigl((nm^2+n^2m)\log(mn)\bigr)$ time.  Alt and Godau~\cite{AG1995} improved the running time to $O(mn\log(mn))$ in their seminal work. Eiter and Mannila~\cite{eiter1994computing} showed that the discrete Fr\'{e}chet distance can be computed in $O(mn)$ time by a simple dynamic programming algorithm.  For roughly two decades, these running times remained state-of-the-art, and the quadratic time barrier stood unbroken.  

Agarwal~et~al.~\cite{agarwal2014computing} eventually broke the barrier for $\tilde{d}_F$ in $\mathbb{R}^2$ in 2013 and gave an algorithm that runs in $O(mn\log\log n/\log n)$ time on a word RAM machine with $\Theta(\log n)$ word size.  They asked whether one can extend this result to compute $d_F$ in subquadratic time. Buchin~et~al.~\cite{buchin2014four} developed a faster \emph{randomized} algorithm for computing $d_F$ in $\mathbb{R}^2$ that takes $O\bigl(mn\sqrt{\log n}(\log\log n)^{3/2}\bigr)$ expected time on a pointer machine and $O\bigl(mn(\log\log n)^2\bigr)$ expected time on a word RAM machine of $\Theta(\log n)$ word size. 
Blank~and~Driemel~\cite{blank2024faster} presented an algorithm in one dimension that computes $d_F$ in $O(m^2\log^2n+n\log n)$ time, which is subquadratic when $m = o(n)$. 



Bringmann proved two important negative results in $\real^d$ for $d \geq 2$~\cite{bringmann2014walking}.  First, the widely accepted \emph{Strong Exponential Time Hypothesis} (SETH) would be refuted if there is a \emph{strongly subquadratic} algorithm for $d_F$ or $\tilde{d}_F$, i.e., a running time of $O\bigl((mn)^{1-\varepsilon}\bigr)$ for any fixed $\varepsilon \in (0,1)$.  
Second, no strongly subquadratic algorithm can approximate $d_F$ within a factor $1.001$ unless SETH fails.  Later, Bringmann and Mulzer~\cite{BW2015} proved that no strongly subquadratic algorithm can approximate $\tilde{d}_F$ within a factor 1.399 even in $\mathbb{R}$ unless SETH fails.  Buchin~et~al.~\cite{buchin2019seth} showed that no strongly subquadratic algorithm can approximate $d_F$ within a factor 3 even in $\mathbb{R}$ unless SETH fails.  

On the positive side, assuming that $m = n$, 
Bringmann and Mulzer~\cite{BW2015} presented an $O(\alpha$)-approximation algorithm for $\tilde{d}_F$ that runs in $O(n\log n + n^2/\alpha)$ time for any $\alpha \in [n]$.  Chan and Rahmati~\cite{chan2018improved} improved the running time to $O(n\log n + n^2/\alpha^2)$ for any $\alpha \in [1,\sqrt{n/\log n}]$.  Later, assuming that $m = n$, Colombe and Fox~\cite{colombe2021approximating} showed an $O(\alpha)$-approximation algorithm for $d_F$ that runs in $O((n^3/\alpha^2)\log n)$ time for any $\alpha \in [\sqrt{n}, n]$.  There is a recent improvement by Van~der~Horst~et~al.~\cite{van2023subquadratic}: assuming that $m \leq n$, there is an $O(\alpha)$-approximation algorithm for $d_F$ that runs in $O((n+mn/\alpha)\log^3n)$ time for any $\alpha \in [1,n]$.  More recently, the running time has been improved to $O((n+mn/\alpha)\log^2n)$ by Van~der~Horst and Ophelders~\cite{vanderhorst_et_al:LIPIcs.SoCG.2024.63}. In all, for any fixed $\varepsilon \in (0,1]$, there are strongly subquadratic $n^\varepsilon$-approximation algorithms for $d_F$ and $\tilde{d}_F$. There are better results for some restricted classes of input curves~\cite{AKW2004,aronov2006frechet,bringmann2015improved,driemel2012approximating,gudmundsson2018fast}.

It is curious that the quadratic time barrier for $d_F$ still stands because several quadratic time barriers in computational geometry have been overcome, including 3SUM and some 3SUM-hard problems~\cite{chan2018,F2017,GS2017,GP2018,KLM2018}, the computation of the discrete Fr\'{e}chet distance~\cite{agarwal2014computing}, and the computation of dynamic time warping and the geometric edit distance~\cite{gold2018dynamic}. In particular, the works in~\cite{agarwal2014computing,gold2018dynamic} are about distances between polygonal curves.
In this paper, we give a subquadratic-time algorithm that computes $d_F$ in $\real^d$ for any fixed $d$.  We use the real RAM model in which a word can store a real number or an integer/pointer of $O(\log n)$ bits.  We assume that standard arithmetic operations can be performed in $O(1)$ time on real numbers and integers.  We also assume that one can use an integer of $O(\log n)$ bits to do a table lookup in $O(1)$ time.     Similar computation models have been used in~\cite{chan2018,gold2018dynamic}.
About the word RAM model used in~\cite{agarwal2014computing,buchin2014four}, the authors also exploit the constant-time table lookup capability only without using bit operations.

For any fixed $d$, our algorithm computes $d_F(\sigma,\tau)$ in $\real^d$ in $O(mn(\log\log n)^{2+\mu}\log n/\log^{1+\mu} m)$ expected time for some constant $\mu \in (0,1)$.  This is the first algorithm that runs in $o(mn)$ time when $m = \Omega(n^{\eps})$ for some fixed $\eps \in (0,1]$.  The computations of $d_F$, $\tilde{d}_F$, and dynamic time warping reduce to filling a two-dimensional dynamic programming table.  We start with the Four-Russians trick that groups cells of the tables into boxes and batch the processing of boxes; this approach has been applied successfully to speed up the computation of $\tilde{d}_F$ and the dynamic time warping distance.  It involves encoding different configurations of these boxes during dynamic programming.  

Our first contribution is using algebraic geometric tools to prove that the number of distinct encodings is subquadratic. This bound is the cornerstone of our result.
Nevertheless, this bound alone does not lead to a subquadratic-time algorithm for computing $d_F$ due to the continuous nature of $d_F$.  Our second contribution is a second level of batch processing that clusters the boxes to save the time for information transfer among the boxes within a cluster.  
This two-level strategy gives a decision algorithm that runs in
$O(mn(\log\log n)^{2+\mu}/\log^{1+\mu} m)$ expected time for some constant $\mu \in (0,1)$.  It is faster than the decision algorithm by Buchin~et~al.~\cite{buchin2014four} for the word RAM model.

Alt~and~Godau~\cite[Theorem 6]{AG1995} showed that if there is a decision procedure that runs in $T(n)$ time, one can compute $d_F(\sigma,\tau)$ in $O(T(n)\log(mn))$ time. Hence, our decision procedure implies that one can compute $d_F(\sigma,\tau)$ in $O(mn(\log\log m)^{2+\mu}\log n/\log^{1+\mu} m)$ expected time.  Algebraic geometric tools have recently been applied to some Fr\'{e}chet and Hausdorff distance problems, including range searching, VC dimension, curve simplification, nearest neighbor, and distance oracle~\cite{afshani2018complexity,bruning2023simplifiedimprovedboundsvcdimension,cheng2023solving,driemel2021vc}.  This paper lends further support to the versatility of this framework for curve proximity problems.

\section{Background}\label{sec: pre}

For any two points $p,q\in \tau$, if $p$ does not appear behind $q$ along $\tau$, we use $\tau[p,q]$ to denote the subcurve of $\tau$ from $p$ to $q$. We use $pq$ to denote the oriented line segment from $p$ to $q$.

\subsection{Fr\'{e}chet distance}
\label{sec:back-frechet}

The key concept in determining $d_F(\sigma,\tau)$ is the \emph{reachability interval}~\cite{AG1995}. Take a point $x\in \tau$ and a point $y\in \sigma$. The pair $(x,y)$ is \emph{reachable} with respect to $\delta$ if and only if $d_F(\tau[v_1, x], \sigma[w_1, y])\le \delta$. The reachability intervals of $x$ include all points in $\sigma$ that can form reachable pairs with $x$. The reachability intervals of $y$ are defined analogously.

Any point $x \in \tau$ or $y \in \sigma$ has at most one reachability interval on every edge of $\sigma$ or $\tau$, respectively. Let $R_{i}[j]=[\ell_{i,j}, r_{i,j}]$ be the reachability interval of the vertex $v_i$ on the edge $w_jw_{j+1}$.  Let $R'_{j}[i]=[\ell'_{j,i}, r'_{j,i}]$ be the reachability interval of the vertex $w_j$ on the edge $v_iv_{i+1}$. 
If a reachability interval is empty, its endpoints are taken to be null.  Alt and Godau~\cite{AG1995} observed that $R_{i+1}[j]$ and $R'_{j+1}[i]$ can be determined in $O(1)$ time using $R_{i}[j]$, $R'_{j}[i]$, and the knowledge of $\tau$ and $\sigma$.  
It is useful to view this computation as filling an $(n-1) \times (m-1)$ dynamic programming table, which has one cell for every pair of edges of $\tau$ and $\sigma$, row by row from the bottom to the top.  Consider the $(i,j)$-cell of this table.  Its lower and left sides are associated with the reachability intervals $R_i[j]$ and $R'_j[i]$, respectively.  The dynamic programming process computes $R_{i+1}[j]$ and $R'_{j+1}[i]$ which are associated with the upper and right sides of the cell, respectively. The algorithm takes $O(mn)$ time to obtain $R_{n}[m-1]$. We have $d_F(\tau, \sigma)\le \delta$ if and only if $w_m\in R_{n}[m-1]$.  We sketch the filling of the dynamic programming table below.

Let $\+B_x$ denote the ball centered at $x$ of radius $\delta$. Let $s_{i,j}$ and $e_{i,j}$ be the start and end of the oriented segment $\+B_{v_i}\cap w_jw_{j+1}$, respectively. Let $s'_{j,i}$ and $e'_{j,i}$ be the start and end of the oriented segment $\+B_{w_j}\cap v_iv_{i+1}$, respectively.

We first initialize the $R_1[j]$'s. If $s_{1,1}=w_1$, set $\ell_{1,1} = s_{1,1}$ and $r_{1,1} = e_{1,1}$; otherwise, both $\ell_{1,1}$ and $r_{1,1}$ are null. For all $j \in [2,m-1]$, if $r_{1,j-1} = s_{1,j}=w_j$, set $\ell_{1,j} = s_{1,j}$ and $r_{1,j} = e_{i,j}$; otherwise, set $\ell_{1,j}$ and $r_{1,j}$ to be null.  The $R'_1[i]$'s are similarly initialized.

Inductively, suppose that we have computed $R_i[j]$ and $R'_j[i]$, and we want to compute $R_{i+1}[j]$ and $R'_{j+1}[i]$.  For any point $y \in w_jw_{j+1}$, a matching between $\tau[v_1, v_{i+1}]$ and $\sigma[w_1, y]$ either matches $v_i$ to some point in $w_jw_{j+1}$ or matches $w_j$ to some point in $v_iv_{i+1}$.  Similarly, for any point $x \in v_iv_{i+1}$, a matching between $\tau[v_1,x]$ and $\sigma[w_1, w_{j+1}]$
either matches $w_j$ to some point in $v_iv_{i+1}$ or matches $v_i$ to some point in $w_jw_{j+1}$.  Hence, if $R_i[j]$ and $R'_j[i]$ are empty, then $R_{i+1}[j]$ and $R'_{j+1}[i]$ are empty. 

Suppose that $R'_j[i]\not=\emptyset$.  There exists a point $x\in v_iv_{i+1}$ such that $d_F(\tau[v_1,x], \sigma[w_1,w_j])\le\delta$.  So $d(x,w_j) \leq \delta$.  For every point $y \in \+B_{v_{i+1}} \cap w_jw_{j+1}$, 
we generate a matching between $\tau[v_1, v_{i+1}]$ and $\sigma[w_1,y]$ that consists of a Fr\'echet matching between $\tau[v_1, x]$ and $\sigma[w_1,w_j]$ and a linear interpolation between $xv_{i+1}$ and $w_jy$.  This matching realizes a distance at most $\delta$, so $d_F(\tau[v_1, v_{i+1}],\sigma[w_1, y])\le \delta$.  In other words, $R_{i+1}[j] = \+B_{v_{i+1}} \cap w_jw_{j+1}$, i.e., $\ell_{i+1,j} = s_{i+1,j}$ and $r_{i+1,j} = e_{i+1,j}$.

Suppose that $R'_j[i]=\emptyset$ and $R_i[j]\not=\emptyset$. In this case, there exists a point $y \in w_jw_{j+1}$ such that $d_F(\tau[v_1,v_i],\sigma[w_1,y]) \leq \delta$.  So $d(v_i,y) \leq \delta$.  Since $R'_j[i] = \emptyset$, the vertex $w_j$ cannot be matched to any point in $v_iv_{i+1}$. 
It implies that a Fr\'echet matching between $\tau[v_1, v_{i+1}]$ and $\sigma[w_1, y]$ matches both $v_i$ and $v_{i+1}$ to points in $w_jw_{j+1}$. Thus, the start $\ell_{i+1,j}$ of $R_{i+1}[j]$ is the first point in the oriented segment $s_{i+1,j}e_{i+1,j} = \+B_{v_{i+1}} \cap w_jw_{j+1}$ that does not lie in front of $\ell_{i,j}$.  If such a point does not exist, then $\ell_{i+1,j} = \text{null}$.  The end $r_{i+1,j}$ of $R_{i+1}[j]$ is $e_{i+1,j}$ if $\ell_{i+1,j} \not= \text{null}$; otherwise, $r_{i+1,j}$ is null.

We can compute $R'_{j+1}[i]$ in a similar way.  The next result follows from the discussion above.

\begin{lemma}~\label{obs: end-ri} 
	If $R_i[j] \not= \emptyset$, then $r_{i,j} = e_{i,j}$ and $\ell_{i,j}\in \{\ell_{i', j}, s_{i'+1,j}, s_{i'+2, j},\ldots, s_{i, j}\}$ for all $i' \in [i-1]$.  Similarly, if $R'_j[i] \not= \emptyset$, then $r'_{j,i}=e'_{j,i}$ and $\ell'_{j,i}\in \{\ell'_{j',i}, s'_{j'+1, i}, s'_{j'+2,i}, \ldots, s'_{j,i}\}$ for all $j' \in [j-1]$.
\end{lemma}


\subsection{Algebraic geometry}  

Given a set $\mathcal{P} = \{\rho_1,...,\rho_s\}$ of polynomials in $\omega$ real variables, a \emph{sign condition vector }$S$ for $\mathcal{P}$ is a vector in $\{-1, 0, +1\}^s$.  The point $\nu \in \mathbb{R}^{\omega}$ \emph{realizes} $S$ if $(\text{sign}(\rho_1(\nu)),...,\text{sign}(\rho_s(\nu))) = S$.   The \emph{realization} of $S$ is the subset $\{\nu \in \real^\omega : \nu\text{ realizes }S\}$.  For every $i \in [s]$, $\rho_i(\nu) = 0$ describes a hypersurface in $\real^\omega$.  The hypersurfaces $\{\rho_i(\nu) = 0 : i \in [s]\}$ partition $\real^\omega$ into open cells of dimensions from 0 to $\omega$.  This set of cells together with the incidence relations among them form an \emph{arrangement} that we denote by $\mathscr{A}(\mathcal{P})$.  Each cell is a connected component of the realization of a sign condition vector for $\mathcal{P}$; the cells in $\mathscr{A}(\mathcal{P})$ represent all realizable sign condition vectors.  

\begin{theorem}[\cite{Basu1995OnCA,pollack1993number}]
	\label{thm:arr}
	Given a set $\mathcal{P}$ of $s$ polynomials of constant degrees in $\omega$ variables, there are $s^\omega \cdot O(1)^{\omega}$ cells in $\mathscr{A}(\+P)$.  One can compute in $s^{\omega+1} \cdot O(1)^{O(\omega)}$ time a set $Q$ of points and the sign condition vectors at these points such that $Q$ contains at least one point in each cell of $\mathscr{A}(\mathcal{P})$.
\end{theorem}

We will need a data structure for $\mathscr{A}(\mathcal{P})$ so that for any query point $\nu \in \real^\omega$, we can determine implicitly the sign condition vector at $\nu$.  An easy way is to linearize the zero sets of the polynomials in $\+P$ to equations of hyperplanes in higher dimensions and then use a point location data structure for an arrangement of hyperplanes.  Let $x_1,\ldots,x_\omega$ denote the variables in $\+P$, and let $D$ denote the maximum degree.   Let $\Phi = \{(e_1,\ldots,e_\omega) : \forall\, i \in [\omega], \; e_i \in \mathbb{Z}_{\ge 0} \, \wedge \, \sum_{i \in [\omega]} e_i \leq D\}$.  A polynomial in $\+P$ is $\sum_{(e_1,\ldots,e_\omega) \in \Phi} c(e_1,\ldots,e_\omega) \cdot \prod_{i \in [\omega]} x_i^{e_1}$ for some coefficient $c(e_1,\ldots,e_\omega)$.  Linearizing it means introducing a new variable $y(e_1,\ldots,e_\omega)$ for each $(e_1,\ldots,e_\omega) \in \Phi$ and rewriting the polynomial as $\sum_{(e_1,\ldots,e_\omega) \in \Phi} c(e_1,\ldots,e_\omega) \cdot y(e_1,\ldots,e_\omega)$.  
Let $M$ be the maximum number of variables in a polynomial in $\+P$. 
There are $O(\omega^M)$ subsets of $\{x_1,\ldots,x_\omega\}$ that have size $M$ or less.  There are less  than $(D+1)^M$ ways to assign exponents to the variables in each subset.  Therefore, the number of new variables is less than $O(\omega^M) \cdot (D+1)^M = O(\omega^{O(1)})$ if both $D$ and $M$ are constants.

\begin{lemma}
	\label{lem:linear}
	Let $\+P$ be a set of polynomials of constant degrees and sizes in $\omega$ variables.  The zero set of each polynomial in $\+P$ can be linearized to a hyperplane equation that has constant size.  The resulting hyerplanes reside in $O(\omega^{O(1)})$ dimensions.
\end{lemma}

We use the point location data strucutre by Erza~et~al.~\cite{ezra2020decomposing}.  Its query time has a big-O~constant that is independent of the dimension.

\begin{theorem}[\cite{ezra2020decomposing}]
	\label{thm:locate}
	Let $\mathcal{H}$ be a set of $s$ hyperplanes in $\real^\omega$.  For any $\eps > 0$, one can construct a data structure that locates any query point in $\mathscr{A}(\mathcal{H})$ in $O(\omega^4\log s)$ time. The space and preprocessing time are $O(s^{2\omega\log\omega+O(\omega)})$.
\end{theorem}

\section{Overview of the decision procedure}
\label{sec: outline}

Let $\alpha = \Theta(\log m/\log\log m)$ be an integer.  Divide $[n]$ into $(n-1)/\alpha$ \emph{blocks} $B_k= [a_k, a_{k+1}]$ for $k \in [(n-1)/\alpha]$, where $a_k =(k-1)\alpha+1$.  For example, $B_1 = (1,\ldots,\alpha+1)$ and $B_2 = (\alpha+1,\ldots,2\alpha+1)$.  Note that $B_{k-1} \cap B_k = \{a_k\}$.  $B_k$ represents the $\alpha$ edges $v_{a_k}v_{a_k+1}, v_{a_k+1}v_{a_k+2},\ldots,v_{a_{k+1}-1}v_{a_{k+1}}$.  

We divide $[m]$ into blocks in a different way.  Let $\theta = \alpha^{\mu}$ for a constant $\mu \in (0,1)$ that will be specified later. Divide $[m]$ into $(m-1)/\theta$ blocks $B'_l = [b_l,b_{l+1}]$ for $l \in [(m-1)/\theta]$, where $b_l = (l-1)\theta+1$. Note that $B'_{l-1}\cap B'_l=\{b_l\}$. $B'_l$ represents the $\theta$ edges $w_{b_l}w_{b_l+1},\ldots,w_{b_{l+1}-1}w_{b_{l+1}}$.

Split the $(n-1) \times (m-1)$ dynamic programming table into $\alpha \times \theta$ boxes that are specified as $(B_k,B'_l)$ for some $k \in [(n-1)/\alpha]$ and $l \in [(m-1)/\theta]$.  We extend the definitions of $R_i[j]$ and $R'_j[i]$ as follows.  Let $R_i[j,j']$ be the set of reachability intervals of $v_i$ on the edges $w_jw_{j+1},\ldots, w_{j'}w_{j'+1}$.  Let $R'_j[i, i']$ be the set of reachability intervals of $w_j$ on the edges $v_iv_{i+1},\ldots,v_{i'}v_{i'+1}$.

We use two encoding schemes. The first is the \emph{signature} of a box $(B_k,B'_l)$ which encodes some useful information about $\tau[v_{a_k}, v_{a_{k+1}}]$ and $\sigma[w_{b_l}, w_{b_{l+1}}]$. A box signature contains $\alpha+\theta$ integers, and we will show that all box signatures can be stored in a table of size $o(mn)$.  We can thus refer to a box signature by its index. The second encoding is an encoding of reachability intervals. The encodings of $R_{a_k}[b_l, b_{l+1}-1]$ and $R'_{b_l}[a_k, a_{k+1}-1]$ constitute the \emph{input encoding} of the box $(B_k,B'_l)$. The encodings of $R_{a_{k+1}}[b_l, b_{l+1}-1]$ and $R'_{b_{l+1}}[a_k, a_{k+1}-1]$ constitute the \emph{output encoding} of the box $(B_k,B'_l)$.  We will show that each input or output encoding can be stored as one integer.

We will show how to construct the output encoding based on the input encoding and the box signature only.  Therefore, the idea is to construct a table in preprocessing so that we can look up the output encoding of a box using a key formed by the index of its signature and its input encoding.  The reachability information can thus be propagated in $O(1)$ time to the next box.  For this approach to work, we need to show that the number of distinct combinations of signatures and input encodings of boxes is subquadratic. We will prove such a bound using polynomials of constant degrees in $\theta$ variables. We also need fast point location data structures for arrangements of hyperplanes to efficiently retrieve the signature and input encoding of a given box.

Unfortunately, the above approach is still not fast enough.
There are two bottlenecks.  First, we need to spend $\Omega(\alpha)$ time to unpack the output encoding of one box and produce the input encoding of the next box.  Second, we need to spend $\Omega(\alpha)$ time to associate each box with its signature.  There are $\Theta(mn/(\alpha\theta))$ boxes, which makes the decision algorithm run in $\Omega(mn/\theta)$ time.  The time to compute $d_F$ is thus $\Omega(mn\log(mn)/\theta) = \Omega(mn)$.
To tackle these bottlenecks, we introduce a second level of batch processing of clusters of boxes.  The idea is to retrieve the signatures of all boxes in a cluster as a batch so that the average time per box is smaller.  Similarly, we can save the time for unpacking the output encodings and producing input encodings of boxes that lie within a cluster.

\section{Encoding schemes}
\label{sec: encoding}

For the rest of the paper, we will assume that $m = \Omega(n^\eps)$ for some fixed $\eps \in (0,1]$ whenever we say that some computation takes $o(mn)$ time.

\subsection{Box signature}
\label{sec:box-sig}

Consider the box $(B_k,B'_l)$. Its signature consists of the signatures of its rows and columns.  Take the column $j$ of the dynamic programming table for some $j \in [b_l,b_{l+1}-1]$.  Let $\+I_j$ be the sequence of points $(s_{a_k,j}, e_{a_k,j}, \ldots, s_{a_{k+1},j}, e_{a_{k+1},j})$.  As discussed in Section~\ref{sec:back-frechet}, the sorted order of $\+I_j$ from $w_j$ to $w_{j+1}$ along $w_jw_{j+1}$ helps deciding $R_{i+1}[j]$ from $R_i[j]$.  The signature of column $j$ is an array $\phi$ of size $2(a_{k+1}-a_k+1)$.  For every $i \in B_k$, $\phi[2i-2a_k+1]$ and $\phi[2i-2a_k+2]$ store the ranks of $s_{i,j}$ and $e_{i,j}$, respectively, when $\+I_j$ is ordered from $w_j$ to $w_{j+1}$ along $w_jw_{j+1}$.  If $s_{i,j}$ and $e_{i,j}$ are null, $\phi[2i-2a_k+1]$ and $\phi[2i-2a_k+2]$ are equal to $2\alpha+3$.   An example is given in Figure~\ref{fig: signature2}.  The signature of row $i$ for some $i \in [a_k,a_{k+1}-1]$ is defined analogously by sorting the $\+I'_i = (s'_{b_l,i}, e'_{b_l,i},\ldots,s'_{b_{l+1},i},e'_{b_{l+1},i})$ along the edge $v_iv_{i+1}$.  The signature of a row or column requires $O(\alpha\log\alpha) = O(\log m)$ bits, so it can be stored as an integer.  The box signature consists of the signatures of $\alpha$ rows and $\theta$ columns.

Computing a box signature as described above requires $O(\alpha\theta\log\alpha)$ time.  We cannot afford it during dynamic programming as there are $O(mn/(\alpha\theta))$ boxes.  Also, as the box signature consists of $\alpha + \theta$ integers, we will not be able to combine it with the input encoding of a box to form a key that can be stored in a word.  This is contrary to the high-level approach sketched in Section~\ref{sec: outline}.  Fortunately, we will show in Section~\ref{sec:preprocess} that there are $o(mn)$ distinct box signatures.  So we can store them in a table of $o(mn)$ size and represent a box signature by its table index.  We will construct a two-dimensional array whose $(k,l)$-entry stores the index of the signature of the box $(B_k,B'_l)$.

\begin{figure}[h]
	\centerline{\includegraphics[scale=0.6]{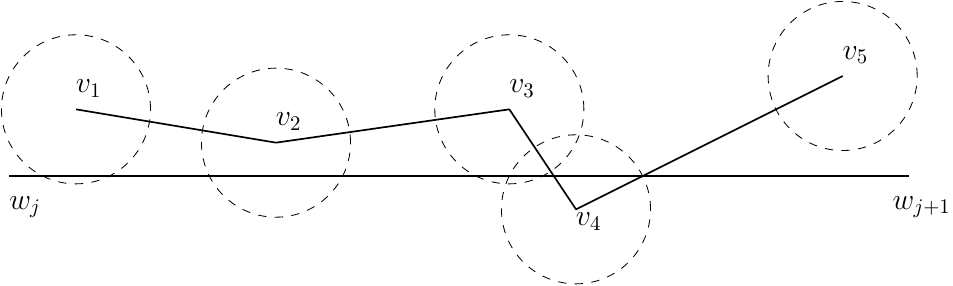}}
	\caption{Suppose that $B_k = [1,5]$ and all balls in this figure have radius $\delta$. The signature of column $j$ is $\phi = (1, 2, 3, 4, 5, 7, 6, 8, 11, 11)$.}
	\label{fig: signature2}
\end{figure}

\subsection{Encoding of reachability intervals}

The reachability interval $R_i[j]$ is encoded with respect to a block $B_k$ that contains $i$.  We denote the encoding by $\mathrm{code}_{B_k}(R_i[j])$.  Note that $R_{a_{k+1}}[j]$ may have different encodings $\mathrm{code}_{B_k}(R_{a_{k+1}}[j])$ and $\mathrm{code}_{B_{k+1}}(R_{a_{k+1}}[j])$ with respect to $B_k$ and $B_{k+1}$, respectively.  Similarly, $R'_j[i]$ is encoded with respect to a block $B'_l$ that contains $j$, and we denote the encoding by $\mathrm{code}_{B'_l}(R'_j[i])$.  

We use $\mathrm{code}_{B_k}(R_i[b_l,b_{l+1}-1])$ to denote $\mathrm{code}_{B_k}(R_i[j])$ for all $j \in [b_l,b_{l+1}-1]$.  The shorthand $\mathrm{code}_{B'_l}(R'_j[a_k,a_{k+1}-1])$ is defined analogously.  The \emph{input encoding} of the box $(B_k,B'_l)$ consists of $\mathrm{code}_{B_k}(R_{a_k}[b_l,b_{l+1}-1])$ and $\mathrm{code}_{B'_l}(R'_{b_l}[a_k,a_{k+1}-1])$.  The \emph{output encoding} of the box $(B_k,B'_l)$ consists of $\mathrm{code}_{B_k}(R_{a_{k+1}}[b_l,b_{l+1}-1])$ and $\mathrm{code}_{B'_l}(R'_{b_{l+1}}[a_k,a_{k+1}-1])$.

Consider a block $B_k$ and some $i \in B_k$.  As discussed in Section~\ref{sec:back-frechet}, if $R_i[j] = \emptyset$, both $\ell_{i,j}$ and $r_{i,j}$ are null; otherwise, $r_{i,j} = e_{i,j}$.  Therefore, it suffices to encode $\ell_{i,j}$, which belongs to $(\ell_{a_k,j},s_{a_k,j},s_{a_k+1,j}, \ldots, s_{a_{k+1},j})$ by Lemma~\ref{obs: end-ri}.  We define $\mathrm{code}_{B_k}(R_i[j])$ to be a 3-tuple $(\pi_{i,j},\beta_{i,j},\gamma_{i,j}) \in [2\alpha+3] \times \{0,1\} \times [0,\alpha+1]$ with the following meaning.

If $R_i[j] = \emptyset$, then $\pi_{i,j} = 2\alpha+3$, signifying that $\ell_{i,j}$ is null; the values of $\beta_{i,j}$ and $\gamma_{i,j}$ are irrelevant in this case.  Suppose that $R_i[j] \not= \emptyset$.  Then, $\gamma_{i,j}$ tells us which element in $(\ell_{a_k,j},s_{a_k,j},\ldots,s_{a_{k+1},j})$ is $\ell_{i,j}$.  We need not store the sequence $(\ell_{a_k,j},s_{a_k,j},\ldots,s_{a_{k+1},j})$ explicitly.  If $\gamma_{i,j} \in [\alpha+1]$, then $\ell_{i,j} = s_{a_k+\gamma_{i,j}-1,j}$.   If $\gamma_{i,j} = 0$, then $\ell_{i,j} = \ell_{a_k,j}$, and the $\pi_{i,j}$-th element in the sorted order of $\{s_{a_k,j}, e_{a_k,j}, \ldots, s_{a_{k+1},j}, e_{a_{k+1},j}\}$ along $w_jw_{j+1}$ is the immediate predecessor of $\ell_{a_k,j}$.  Also, $\beta_{i,j} = 1$ if and only if $\ell_{a_k,j}$ is equal to this immediate predecessor.

The encoding $\mathrm{code}_{B'_l}(R'_j[i])$ is defined analogously as a 3-tuple $(\pi'_{j,i}, \beta'_{j,i}, \gamma'_{j,i})$.  The input or output encoding of a box uses $O(\alpha\log\alpha) = O(\log m)$ bits, so it can be stored as an integer.

Suppose that $i+1 \in B_k$.  We describe how to produce $\mathrm{code}_{B_k}(R_{i+1}[j])$ using the signature of $(B_k,B'_l)$, $\mathrm{code}_{B_k}(R_i[j])$, and $\mathrm{code}_{B'_l}(R_j[i])$.  As described in Section~\ref{sec:back-frechet}, we need to decide whether $R_i[j]$ or $R'_j[i]$ is null, which can be done by checking whether $\pi_{i,j}$ or $\pi'_{j,i}$ is equal to $2\alpha+3$.  Next, we need to compare $\ell_{i,j}$ with $s_{i+1,j}$ and $e_{i+1,j}$.  There are two cases in handling this comparison.

Suppose that $\gamma_{i,j} \in [\alpha+1]$.   We have $\ell_{i,j} = s_{a_k + \gamma_{i,j}-1,j}$ in this case.  We use the signature $\phi$ of column $j$ to retrieve $\phi[2\gamma_{i,j}-1]$, which is the rank of $s_{a_k+\gamma_{i,j}-1,j}$ among the points in $\+I_j$ ordered from $w_j$ to $w_{j+1}$ along $w_jw_{j+1}$.  Comparing $\phi[2\gamma_{i,j}-1]$ with $\phi[2(i+1)-2a_k+1]$ and $\phi[2(i+1)-2a_k+2]$ give the results of comparing $\ell_{i,j}$ with $s_{i+1,j}$ and $e_{i+1,j}$, which allow us to determine $\ell_{i+1,j}$ as described in Section~\ref{sec:back-frechet}. If $\ell_{i+1,j}$ is not null, it is equal to either $s_{a_k+\gamma_{i,j}-1,j}$ or $s_{i+1,j}$.  We thus construct $\mathrm{code}_{B_k}(R_{i+1}[j])$ as follows. If $\ell_{i+1,j}$ is null, set $\pi_{i+1,j} = 2\alpha+3$; if $\ell_{i+1,j} = s_{a_k+\gamma_{i,j}-1,j}$, set $\gamma_{i+1,j} = \gamma_{i,j}$; if $\ell_{i+1,j} = s_{i+1,j}$, set $\gamma_{i+1,j} = i-a_k+2$.

Suppose that $\gamma_{i,j} = 0$.  So $\ell_{i,j} = \ell_{a_k,j}$. Given the signature $\phi$ of column $j$, we get the rank $\phi[\pi_{i,j}]$ of the immediate predecessor of $\ell_{i,j}$ among the points in $\+I_j$ sorted along $w_jw_{j+1}$.  We use $\beta_{i,j}$ and the comparsions of $\phi[\pi_{i,j}]$ with $\phi[2(i+1)-2a_k+1]$ and $\phi[2(i+1)-2a_k+2]$ to determine the comparsions of $\ell_{a_k,j}$ with $s_{i+1,j}$ and $e_{i+1,j}$.  This tells us whether $\ell_{i+1,j}$ is null, $\ell_{a_k,j}$, or $s_{i+1,j}$ as described in Section~\ref{sec:back-frechet}.  We construct $\mathrm{code}_{B_k}(R_{i+1}[j])$ as follows. If $\ell_{i+1,j}$ is null, set $\pi_{i+1,j} = 2\alpha+3$; if $\ell_{i+1,j} = \ell_{a_k,j}$, set $\gamma_{i+1,j} = 0$, $\pi_{i+1,j} = \pi_{i,j}$, and $\beta_{i+1,j} = \beta_{i,j}$; if $\ell_{i+1,j} = s_{i+1,j}$, set $\gamma_{i+1,j} = i-a_k+2$.

Hence, we can compute $\mathrm{code}_{B_k}(R_{i+1}[j])$ in $O(1)$ time.  Similarly, if $j+1 \in B'_l$, we can also compute $\mathrm{code}_{B'_l}(R'_{j+1}[i])$ in $O(1)$ time.  In all, given the signature of the box $(B_k,B'_l)$ and its input encoding, we can process the cells in the box row by row to obtain the output encoding.

\begin{lemma}
	\label{lem: encoding}
	Given the signature of a box $(B_k,B'_l)$ and its input encoding, we can compute its output encoding in $O(\alpha\theta)$ time.  The signature consists of $\alpha+\theta$ integers.  The input encoding occupies one integer, and so does the output encoding.
\end{lemma}

\section{Subquadratic decision algorithm}
\label{sec:preprocess}



\subsection{Preprocessing}
\label{sec: preprocessing}

Given an oriented line segment $pq$, we say $x \le_{pq} y$ for two points $x,y \in pq$ if $x$ is not behind $y$ along $pq$.  Consider a box $(B_k, B'_l)$ and the predicates in Table~\ref{tb:pred} over all $i, i' \in B_k$ and $j,j' \in B'_l$.  $P_1$ and $P_2$ capture the distances between the vertices of $\tau$ and the edges of $\sigma$ and vice versa.  For each $j$, $P_4$, $P_6$, and $P_8$ capture the ordering of the endpoints of $\+B_{v_i} \cap w_jw_{j+1}$ for $i \in B_k$ along the edge $w_jw_{j+1}$.  For each $i$, $P_3$, $P_5$, and $P_7$ capture the ordering of the endpoints of $\+B_{w_j} \cap v_iv_{i+1}$ for $j \in B'_l$ along the edge $v_iv_{i+1}$.  Hence, we can determine the signatures of all rows and columns of the box $(B_k,B'_l)$ from the truth values of these predicates without checking $\tau$ and $\sigma$ further.

\begin{table}
	\centerline{\begin{tabular}{lcl}
			\hline
			\addlinespace[1pt]
			$P_1(i,j) = \text{true}$ & $\iff$ &  $d(w_{j}, v_{i}v_{i+1})\leq\delta$ \\[.1em] \hline
			\addlinespace[1pt]
			$P_2(i,j) = \text{true}$ & $\iff$ & $d(v_{i},w_{j}w_{j+1}) \leq \delta$\\[.1em] \hline
			\addlinespace[1pt]
			$P_3(i,j,j') = \text{true}$ & $\iff$ &  $\text{neither $s'_{j, i}$ nor $s'_{j', i}$ is null} \, \wedge \, s'_{j, i}\le_{v_iv_{i+1}}s'_{j', i}$\\[.1em] \hline
			\addlinespace[1pt]
			$P_4(i,i',j) = \text{true}$ & $\iff$ & $\text{neither $s_{i, j}$ nor $s_{i', j}$ is null} \, \wedge \, s_{i, j}\le_{w_jw_{j+1}}s_{i', j}$ \\[.1em] \hline
			\addlinespace[1pt]
			$P_5(i, j, j') = \text{true}$ & $\iff$ & $\text{neither $e'_{j, i}$ nor $e'_{j', i}$ is null} \, \wedge \, e'_{j, i}\le_{v_iv_{i+1}}e'_{j', i}$\\[.1em] \hline
			\addlinespace[1pt]
			$P_6(i, i', j) = \text{true}$ & $\iff$ & $\text{neither $e_{i, j}$ nor $e_{i', j}$ is null} \, \wedge \, e_{i, j}\le_{w_jw_{j+1}}e_{i', j}$\\[.1em] \hline
			\addlinespace[1pt]
			$P_7(i, j, j') = \text{true}$ & $\iff$ & $\text{neither $s'_{j, i}$ nor $e'_{j', i}$ is null} \, \wedge \, s'_{j, i}\le_{v_iv_{i+1}}e'_{j', i}$\\[.1em] \hline
			\addlinespace[1pt]
			$P_8(i, i', j) = \text{true}$ & $\iff$ & $\text{neither $s_{i, j}$ nor $e_{i', j}$ is null} \, \wedge \, s_{i, j}\le_{w_jw_{j+1}}e_{i', j}$ \\[.1em]
			\hline
	\end{tabular}}
	\caption{There are $O(\alpha^2\theta)$ predicates over all $i,i' \in B_k$ and $j,j' \in B'_l$.}
	\label{tb:pred}
\end{table}

\subsubsection{Signature data structures}
\label{sec:sig}

Fix $B_k$.  
Consider a sequence $(w_{\ell},\ldots,w_{\ell+\theta})$. 
The signatures of rows and columns of the box specified by $B_k$ and $[\ell,\ell+\theta]$ are completely determined by the truth values of $P_1,\ldots,P_8$ over all $i,i' \in B_k$ and $j,j \in [\ell,\ell+\theta]$.  By treating $\delta$ and the coordinates of $v_{a_k},\ldots,v_{a_{k+1}}$ as coefficients and the coordinates of $w_\ell,\ldots,w_{\ell+\theta}$ as $d(\theta+1)$ variables, we can formulate a set $\+P_k$ of $O(\alpha^2\theta)$ polynomials of constant degrees and sizes in these $d(\theta+1)$ variables that implement $P_1,\ldots,P_8$ over all $i,i' \in B_k$ and $j,j' \in B'_l$.  Each sign condition vector for $\+P_k$ determines the signature of a box that involves $B_k$.  These polynomials are specified in Appendix~\ref{app:predicate}.  By Theorem~\ref{thm:arr}, there are at most $\alpha^{O(\theta)}$ sign condition vectors of $\+P_k$, and we can compute them all in $\alpha^{O(\theta)}$ time.\footnote{It is customary to write $O(\alpha^{O(\theta)})$.  To simplify the exposition, when we have a term $t^{O(h)}$ inside the big-Oh such that $t$ is an increasing function in $m$ or $n$ and $h \geq 1$, we can increase the hidden constant in the exponent $O(h)$ to absorb the big-Oh constant.  We will do without the big-Oh this way.}  We assign unique indices from the range $[\alpha^{O(\theta)}]$ to these sign condition vectors.  For each sign condition vector $\nu$, we compute the corresponding box signature and store it in the array entry {\sc Signa}$[k,\mathit{id}]$, where $\mathit{id}$ is the index of $\nu$.  Repeating the above for all $k \in [(n-1)/\alpha]$ produces the array {\sc Signa} in $n\alpha^{O(\theta)}$~time.

\begin{lemma}
	\label{lem:sig}
	We can compute in $n\alpha^{O(\theta)}$ time the array {\sc Signa}$[1\,..\,(n-1)/\alpha, 1\,..\,\alpha^{O(\theta)}]$.  For every $k \in [(n-1)/\alpha]$, the subarray {\sc Signa}$[k]$ stores the signatures of all boxes that involve $B_k$.
\end{lemma}

\subsubsection{Reachability propagation and indexing}
\label{sec:structures}

We want to construct a table {\sc Map} so that for any index \emph{key} that gives the signature of a box and its input encoding, {\sc Map} can be accessed in $O(1)$ time to return the output encoding for \emph{key}.  The box signatures can be obtained by going through {\sc Signa}.  Since the encoding of a reachability interval is an element of $[2\alpha+3] \times \{0,1\} \times [0,\alpha+1]$, there are $O(\alpha^2)$ possibilities.  There are $\alpha + \theta$ reachability intervals, so there are 
$O(\alpha^2)^{\alpha+\theta} = \alpha^{O(\alpha)}$ possibilities which can be enumerated in $\alpha^{O(\alpha)}$ time.  An index of {\sc Map} consists of $k \in [(n-1)/\alpha]$ and $\mathit{id} \in [\alpha^{O(\theta)}]$ that specify a box signature in {\sc Signa} as well as $O(\alpha\log\alpha)$ bits that represent an input encoding.  Such an index can be stored in $\log n + O(\alpha\log\alpha)$ bits which can fit into one word.  As a result, {\sc Map} can be accessed in $O(1)$ time.  Every entry of {\sc Map} can be computed in $O(\alpha\theta)$ time by Lemma~\ref{lem: encoding}.

\begin{lemma}
	\label{lem:map}
	We can compute {\sc Map} in $n\alpha^{O(\alpha)}$ time and space.
\end{lemma}

We want to construct an array {\sc Index}$[1\,..\,(n-1)/\alpha, 1\,..\,(m-1)/\theta]$ so that {\sc Signa}$[k,\text{\sc Index}[k,l]]$ stores the signature of the box $(B_k,B'_l)$.  We linearize the zero sets of the polynomials in $\+P_k$ to a set $\+H_k$ of hyperplanes.  There are $d(\theta+1)$ variables in $\+P_k$.  By Lemma~\ref{lem:linear}, the linearization gives rise to $O(\theta^L)$ variables for some constant $L > 1$.  By Theorem~\ref{thm:locate}, we can construct a point location data structure for $\mathscr{A}(\+H_k)$ such that a query can be answered in $O(\theta^{4L}\log \alpha)$ time.  For each $l \in [(m-1)/\theta]$, we use the coordinates of $w_{b_l},\ldots,w_{b_{l+1}}$ to compute the values of the variables for $\+H_k$ and then query the point location data structure in $O(\theta^{4L}\log\alpha)$ time. This gives the index $\mathit{id}$ for the sign condition vector of $\+P_k$ induced by $(B_k,B'_l)$. We set {\sc Index}$[k,l] = \mathit{id}$. The above construction takes $\alpha^{O(\theta^L\log\theta)}$ expected time by Theorem~\ref{thm:locate}. Recall that $\theta = \alpha^\mu$. We choose a constant $\mu < 1/L$ so that $\theta^L\log\theta < \alpha$.  If $\alpha = c\log m/\log\log m$ for a small enough constant $c$, the time to construct the point location data structures is $n\alpha^{O(\theta^L\log\theta)} = n\alpha^{O(\alpha)} = O(nm^\eps)$ for some $\eps \in (0,1)$, which is dominated by the point location cost for the entries of {\sc Index}.

\begin{lemma}
	\label{lem:index}
	If $\theta^L\log \theta < \alpha$, we can construct {\sc Index} in $O(mn\theta^{4L-1}\log\alpha/\alpha)$ expected time and space.
\end{lemma}

\subsubsection{$\pmb{B_k}$-coder and $\pmb{B'_l}$-coder} 
\label{sec:coder}

When we reach the box $(B_k,B'_l)$ during dynamic programming, we have read $\mathrm{code}_{B_{k-1}}(R_{a_k}[b_l,b_{l+1}-1])$ and $\mathrm{code}_{B'_{l-1}}(R'_{b_l}[a_k,a_{k+1}-1])$ from {\sc Map}.  To read the output encoding of the box $(B_k,B'_l)$ from {\sc Map}, we must first produce its input encoding $\mathrm{code}_{B_k}(R_{a_k}[b_l,b_{l+1}-1])$ and $\mathrm{code}_{B'_l}(R'_{b_l}[a_k,a_{k+1}-1])$.  The $B_k$-coder and $B'_l$-coder serve this purpose.

Assuming that we have inductively computed $\ell_{a_{k-1},j}$, it is easy to unpack $\mathrm{code}_{B_{k-1}}(R_{a_k}[b_l,b_{l+1}-1])$ and $\mathrm{code}_{B'_{l-1}}(R'_{b_l}[a_k,a_{k+1}-1])$ to $R_{a_k}[b_l,b_{l+1}-1]$ and $R'_{b_l}[a_k,a_{k+1}-1]$ in $O(\alpha\log\alpha)$ time, respectively.  This gives us the coordinates of $\ell_{a_k,j}$ and $r_{a_k,j}$ for $j \in [b_l,b_{l+1}-1]$ and those of $\ell'_{b_l,i}$ and $r'_{b_l,i}$ for $i \in [a_k,a_{k+1}-1]$.  We describe the $B_k$-coder.  The $B'_l$-coder is built similarly.  Given any $R_{a_k}[j]$, the $B_k$-coder is to return a 3-tuple $(\pi_{a_k,j},\beta_{a_k,j},\gamma_{a_k,j})$ as $\mathrm{code}_{B_k}(R_{a_k}[j])$.  Let $(\pi_{a_{k-1},j}, \beta_{a_{k-1},j}, \gamma_{a_{k-1},j})$ be $\mathrm{code}_{B_{k-1}}(R_{a_k}[j])$.   

If $R_{a_k}[j] = \text{null}$, set $\pi_{a_k,j} = 2\alpha+3$ and we are done.  Suppose not.  If $\gamma_{a_{k-1},j} \in [\alpha+1]$, then $\ell_{a_k,j} = s_{a_{k-1}+\gamma_{a_{k-1},j}-1,j}$. 
If $\gamma_{a_{k-1},j} = 0$, then $\ell_{a_k,j} = \ell_{a_{k-1},j}$.  We have inductively computed $\ell_{a_{k-1},j}$.  As a result, we know the coordinates of $\ell_{a_k,j}$ in both cases, and it suffices to locate $\ell_{a_k,j}$ among the sorted order of $s_{a_k,j}, e_{a_k,j}, \ldots, s_{a_{k+1},j}, e_{a_{k+1},j}$ along $w_jw_{j+1}$ to determine $(\pi_{a_k,j}, \beta_{a_k,j}, \gamma_{a_k,j})$. We cannot afford to maintain the sorted sequence for all choices of $w_jw_{j+1}$ explicitly.  Instead, we formulate $O(\alpha)$ polynomials of constant degrees in $3d$ variables to model the positioning of an unknown point among the sorted order of $s_{a_k,j}, e_{a_k,j}, \ldots, s_{a_{k+1},j}, e_{a_{k+1},j}$ along $w_jw_{j+1}$. The $3d$ variables correspond to the coordinates of the unknown point, $w_j$, and $w_{j+1}$.  These polynomials are specified in Appendix~\ref{app:coder}.  We linearize the zero sets of these polynomials to hyperplanes and construct the point location data structure in Theorem~\ref{thm:locate}.  During dynamic programming, we query this data structure with $(\ell_{a_k,j},w_j,w_{j+1})$ to get $\mathrm{code}_{B_k}(R_{a_k}[j])$.  

\begin{lemma}
	\label{lem:coder}
	Each $B_k$-coder can be built in $\alpha^{O(1)}$ space and expected time. Given $R_{a_k}[b_l,b_{l+1}-1]$, computing $\mathrm{code}_{B_k}(R_{a_k}[b_l,b_{l+1}-1])$ takes $O(\theta\log\alpha)$ time.  Each $B'_l$-coder can be built in $\theta^{O(1)}$ space and expected time.  Given $R'_{b'_l}[a_k,a_{k+1}-1]$, computing $\mathrm{code}_{B'_l}(R'_{b_l}[a_k,a_{k+1}-1])$ takes $O(\alpha\log\theta)$~time. 
\end{lemma}

\subsection{Dynamic programming}
\label{sec:compact-dp}

Initialize $R_1[1,m-1]$ and $R'_1[1,n-1]$ in $O(m+n)$ time.   By Lemma~\ref{lem:coder}, we query the $B_1$-coder and the $B'_1$-coder to obtain $\mathrm{code}_{B_1}(R_1[b_l,b_{l+1}-1])$ for all $l \in [(m-1)/\theta]$ and $\mathrm{code}_{B'_1}(R'_1[a_k,a_{k+1}-1])$ for all $k \in [(n-1)/\alpha]$.  The time needed is $O(m\log\alpha + n\log\theta) = O(n\log\alpha)$.

We now have the input encoding of the box $(B_1,B'_1)$. Concatenate 
{\sc Signa}$[1,\text{\sc Index}[1,1]]$, and this input encoding to form an index {\em key}.  {\sc Map}$[\mathit{key}]$ stores the output encoding $\mathrm{code}_{B_1}(R_{a_2}[b_1,b_2-1])$ and $\mathrm{code}_{B'_1}(R'_{b_2}[a_1,a_2-1])$.  Unpack $\mathrm{code}_{B'_1}(R'_{b_2}[a_1,a_2-1])$ to $R'_{b_2}[a_1,a_2-1]$ in $O(\alpha\log\alpha)$ time.  Use the $B'_2$-coder to compute $\mathrm{code}_{B'_2}(R'_{b_2}[a_1,a_2-1])$ in $O(\alpha\log\theta)$ time.  We now have the input encoding of the box $(B_1,B'_2)$.  
We can then repeat to process $(B_1,B'_l)$ for all $l \in [(m-1)/\theta]$.  The time needed is $O(m\alpha\log\alpha/\theta)$.

We now have $\mathrm{code}_{B_1}(R_{a_2}[b_l,b_{l+1}-1])$ for all $l \in [(m-1)/\theta]$.  Unpack them to $R_{a_2}[b_l,b_{l+1}-1]$ and use the $B_2$-coder to compute $\mathrm{code}_{B_2}(R_{a_2}[b_l,b_{l+1}-1])$ for all $l \in [(m-1)/\theta]$. This takes $O((m/\theta) \cdot \theta\log\alpha) = O(m\log\alpha)$ time.  We proceed as before to process $(B_2,B'_l)$ for all $l \in [(m-1)/\theta]$ in $O(m\alpha\log\alpha/\theta)$ time.  There are $(n-1)/\alpha$ rows, hence a total of $O(mn\log\alpha/\theta)$ time.

\begin{lemma}
	\label{lem:DP}
	After the proprocessing in Section~\ref{sec: preprocessing}, computing $R_n[m-1]$ takes $O(mn\log\alpha/\theta)$ time.
\end{lemma}

\section{Further improvement}
\label{sec:faster}

Lemmas~\ref{lem:sig}--\ref{lem:DP} imply that the decision algorithm runs in $O(mn\theta^{4L-1}\log\alpha/\alpha + mn\log\alpha/\theta)$ expected time.  In this section, we improve it to $O(mn\log^2\alpha/(\alpha\theta))$ by proposing a second level batch processing to tackle the bottlenecks in Lemmas~\ref{lem:index} and~\ref{lem:DP}.

\subsection{Signatures}

To deal with the bottleneck in Lemma~\ref{lem:index}, we cluster the boxes into $\alpha \times 1$ disjoint tiles which we call \emph{racks}.  A rack is an $\alpha \times 1$ array of boxes or an $\alpha^2 \times \theta$ array of cells of the dynamic programming table.  Each rack is specified as $(\hat{B}_\kappa,B'_l)$, where $\hat{B}_\kappa$ consists of $B_k,B_{k+1},\ldots,B_{k+\alpha-1}$ for some $k$.

Consider the predicates $P_1,\ldots,P_8$ over all $i,i' \in \hat{B}_\kappa$ and $j, j' \in [\ell,\ell+\theta]$ for some $\ell$. There are $O(\alpha^4\theta)$ of them, and their truth values determine the signatures of the boxes inside the rack $(\hat{B}_\kappa, [\ell,\ell+\theta])$.  By treating the coordinates of $w_{\ell},\ldots,w_{\ell+\theta}$ as variables, we can formulate a set $\+Q_k$ of $O(\alpha^4\theta)$ polynomials of constant degrees and sizes that implement the predicates.

By Theorem~\ref{thm:arr}, we compute in $\alpha^{O(\theta)}$ time all sign condition vectors of $\+Q_\kappa$ and the box signatures associated with them.  Assign unique indices from the range $[\alpha^{O(\theta)}]$ to these sign condition vectors.  Take a sign condition vector $\nu$ of $\+Q_\kappa$.  For each $i \in [\alpha]$, the vector $\nu$ induces one box signature $\xi_{i}$ that involves $B_{k+i-1}$, and we set {\sc Signa}$[k+i-1,\mathit{id}] = \xi_{i}$, where $\mathit{id}$ is the index of $\nu$.

For each $\kappa$, we fill the subarrays {\sc Index}$[k+i-1]$ for all $i\in[\alpha]$ as follows. Linearize $\+Q_\kappa$ to a set $\+H_\kappa$ of hyperplanes.  There are $O(\theta^L)$ variables in $\+H_\kappa$ for some constant $L > 1$.  Construct the point location data structure for $\mathscr{A}(\+H_\kappa)$ in $\alpha^{O(\theta^L\log\theta)}$ expected time.  For each $l \in [(m-1)/\theta]$, use the coordinates of $w_{b_l},\ldots,w_{b_{l+1}}$ to query this data structure in $O(\theta^{4L}\log \alpha)$ time.  It returns a sign condition vector $\nu$.  Let $\mathit{id}$ be the index of $\nu$.  We set {\sc Index}$[k+i-1,l] = \mathit{id}$ for all $i \in [\alpha]$.   

We spend $n\alpha^{O(\theta)}$ time to compute the sign conditions vectors of $\+Q_\kappa$ for all $\kappa$ and the array {\sc Signa}.  We spend $n\alpha^{O(\theta^L\log\theta)} \leq n\alpha^{O(\alpha)}$ expected time to obtain the point location data structures for $\mathscr{A}(\+H_\kappa)$ for all $\kappa$.  We also spend $O(\theta^{4L}\log\alpha)$ to perform the point location for every rack; the total over all racks is $O(mn\theta^{4L}\log\alpha/(\alpha^2\theta)) = O(mn\log\alpha/(\alpha\theta))$, provided that $\theta^{4L} = \alpha^{4\mu L} < \alpha$, i.e., choose $\mu < 1/(4L)$.  Note that $n\alpha^{O(\theta)} \leq n\alpha^{O(\alpha)}$ is dominated by $O(mn\log\alpha/(\alpha\theta))$.

\begin{lemma}
	\label{lem:index2}
	If $\theta^{4L} < \alpha$, we can compute {\sc Signa} and {\sc Index} in $O(mn\log\alpha/(\alpha\theta))$ expected time.
\end{lemma}

\subsection{Input encoding generation}

The bottleneck in Lemma~\ref{lem:DP} arises from unpacking the output encoding of the box $(B_k,B'_l)$ and calling the $B_{k+1}$-coder and the $B'_{l+1}$-coder.  
We cluster the boxes into $\alpha \times \alpha^2/\theta$ tiles which we call \emph{superboxes}.  A superbox is an $\alpha \times \alpha^2/\theta$ array of boxes or an $\alpha^2 \times \alpha^2$ array of cells of the dynamic programming table.  So there are $\Theta(mn/\alpha^4)$ superboxes.  We specify each superbox as $(\hat{B}_\kappa,\hat{B}'_\lambda)$ where $\hat{B}_\kappa$ consists of $B_k, B_{k+1}, \ldots, B_{k+\alpha-1}$ for some appropriate $k$, and $\hat{B}_\lambda$ consists of $B'_l,B'_{l+1},\ldots, B'_{l+\alpha^2/\theta-1}$ for some appropriate $l$.

Our solution is to derive the input and output encodings of all boxes within the superbox $(\hat{B}_\kappa,\hat{B}'_\lambda)$ with respect to $\hat{B}_\kappa$ and $\hat{B}'_\lambda$.  No unpacking and recoding is needed within a superbox.  We only need to perform such unpacking and recoding when the dynamic programming process proceeds from $(\hat{B}_\kappa,\hat{B}'_\lambda)$ to $(\hat{B}_{\kappa+1},\hat{B}'_\lambda)$ and $(\hat{B}_\kappa,\hat{B}'_{\lambda+1})$.  

What does it mean to encode a reachability interval $R_i[j]$ with respect to $\hat{B}_\kappa$?  The encoding is also a 3-tupe $(\hat{\pi}_{i,j}, \hat{\beta}_{i,j}, \hat{\gamma}_{i,j})$ as in the case of boxes.  The difference is that $\hat{\pi}_{i,j}$ and $\hat{\gamma}_{i,j}$ come from bigger domains.  Specifically, $(\hat{\pi}_{i,j}, \hat{\beta}_{i,j}, \hat{\gamma}_{i,j}) \in [2\alpha^2+3] \times \{0,1\} \times [0,\alpha^2+1]$.  
The element $\hat{\gamma}_{i,j}$ tells us which element of $(\ell_{a_k,j}, s_{a_k,j}, \ldots, s_{a_{k+\alpha},j})$ is equal to $\ell_{i,j}$.  If $\hat{\gamma}_{i,j} \in [\alpha^2+1]$, then $\ell_{i,j} = s_{a_k+\hat{\gamma}_{i,j}-1,j}$.  If $\hat{\gamma}_{i,j} = 0$, then $\ell_{i,j} = \ell_{a_k,j}$, and the $\hat{\pi}_{i,j}$-th element in the sequence $(s_{a_k,j},e_{a_k,j}, \ldots, s_{a_{k+\alpha},j}, e_{a_{k+\alpha},j})$ is the immediate predecessor of $\ell_{i,j}$ in the sorted order of this sequence along $w_jw_{j+1}$.  Also, $\hat{\beta}_{i,j} = 1$ if and only if $\ell_{a_\kappa,j}$ is equal to this immediate predecessor.

As a result, we can construct the $\hat{B}_\kappa$-coder as in the construction of a $B_k$-coder in Section~\ref{sec:coder}.  The only difference is that there are $O(\alpha^2)$ polynomials instead of $O(\alpha)$ polynomials.  The construction time remains to be $\alpha^{O(1)}$.  The query time remains to be $O(\log\alpha)$.

The $\hat{B}'_\lambda$-coder is also organized like the $B'_l$-coder.  The difference is that the $B'_l$-coder involves $\theta$ columns and hence $O(\theta)$ polynomials, but the $\hat{B}'_\lambda$-coder involves $\alpha^2$ and hence $O(\alpha^2)$ polynomials like the $\hat{B}_\kappa$-coder.  So the $\hat{B}'_\lambda$ has the same performance as the $\hat{B}_\kappa$-coder.

\begin{lemma}
	\label{lem:coder2}
	Each $\hat{B}_\kappa$-coder and each $\hat{B}'_\lambda$-coder can be built in $\alpha^{O(1)}$ space and expected time. Given $R_{a_\kappa}[b_\lambda,b_{\lambda+1}-1]$, computing $\mathrm{code}_{\hat{B}_\kappa}(R_{a_\kappa}[b_\lambda,b_{\lambda+1}-1])$ takes $O(\alpha^2\log\alpha)$ time.  The same result holds for computing $\mathrm{code}_{\hat{B}'_\lambda}(R'_{b_\lambda}[a_\kappa,a_{\kappa+1}-1])$ from $R'_{b_\lambda}[a_\kappa,a_{\kappa+1}-1]$.
\end{lemma}

When proceeding from $(\hat{B}_\kappa,\hat{B}'_\lambda)$ to $(\hat{B}_{\kappa+1},\hat{B}'_\lambda)$, we unpack $\mathrm{code}_{\hat{B}_\kappa}(R_{a_{\kappa+1}}[b_\lambda-b_{\lambda+1}-1])$ to $R_{a_{\kappa+1}}[b_\lambda-b_{\lambda+1}-1]$ in $O(\alpha^2\log\alpha)$ time.  We compute $\mathrm{code}_{\hat{B}_{\kappa+1}}(R_{a_\kappa}[b_\lambda,b_{\lambda+1}-1])$ in $O(\alpha^2\log\alpha)$ time by Lemma~\ref{lem:coder2}.  
The same can be said for proceeding from $(\hat{B}_\kappa,\hat{B}'_\lambda)$ to $(\hat{B}_\kappa,\hat{B}'_{\lambda+1})$.  
The total over all $O(mn/\alpha^4)$ superboxes is $O(mn\log\alpha/\alpha^2)$.

\begin{lemma}
	\label{lem:DP2}
	Suppose that {\sc Signa}, {\sc Index}, and the $\hat{B}_\kappa$-coder and $\hat{B}'_\lambda$-coder are available for all $\kappa \in [(n-1)/\alpha^2]$ and $\lambda \in [(m-1)/\alpha^2]$.  We can compute $R_n[m-1]$ in $O(mn\log\alpha/\alpha^2)$ time.
\end{lemma}

Putting Lemmas~\ref{lem:map}, \ref{lem:index2}, \ref{lem:coder2}, and \ref{lem:DP2} together gives our main results.  Recall that $\theta = \alpha^\mu$ and $\alpha = \Theta(\log m/\log\log m)$.

\begin{theorem}
	Let $\sigma$ and $\tau$ be two polygonal curves of sizes $m$ and $n$ in $\mathbb{R}^d$ for some fixed $d$.  Assume that $m \leq n$.  There exists a constant $\mu \in (0,1)$ such that {\em (i)}~one can decide whether $d_F(\sigma,\tau) \leq \delta$ in $O(mn(\log\log m)^{2+\mu}/\log^{1+\mu} m)$ expected time, and {\em (ii)}~one can compute $d_F(\sigma,\tau)$ in $O(mn(\log\log m)^{2+\mu}\log n/\log^{1+\mu} m)$ expected time.  Hence, if $m = \Omega(n^\eps)$ for any fixed $\eps \in (0,1]$, one can compute $d_F(\sigma,\tau)$ in $o(mn)$ time.
\end{theorem} 

\section{Conclusion}

A natural question is whether one can shave more log factors from the running time.  Our running time is not deterministic because of the expected construction time bound of the point location data structure in Theorem~\ref{thm:locate}.  It may be possible to design a deterministic construction algorithm.

\newpage

\appendix 

\section{Polynomials for predicates}
\label{app:predicate}

We restate the predicates we are going to implement as follows.
\cancel{
	\begin{itemize}
		\item $P_1(i,j)$ returns true if and only if $d(w_{j}, v_{i}v_{i+1})\leq \delta$. 
		\item $P_2(i,j)$ returns true if and only if $d(v_{i},w_{j}w_{j+1}) \leq \delta$. 
		\item $P_3(i,j,j')$ returns true if and only if both $s'_{j, i}$ and $s'_{j', i}$ are not null, and $s'_{j, i}\le_{v_iv_{i+1}}s'_{j', i}$. 
		\item $P_4(i,i',j)$ returns true if and only if both $s_{i, j}$ and $s_{i', j}$ are not null, and $s_{i, j}\le_{w_jw_{j+1}}s_{i', j}$.
		\item $P_5(i, j, j')$ returns true if and only if both $e'_{j, i}$ and $e'_{j', i}$ are not null, and $e'_{j, i}\le_{v_iv_{i+1}}e'_{j', i}$. 
		\item $P_6(i, i', j)$ returns true if and only if both $e_{i, j}$ and $e_{i', j}$ are not null, and $e_{i, j}\le_{w_jw_{j+1}}e_{i', j}$.
		\item $P_7(i, j, j')$ returns true if and only if both $s'_{j, i}$ and $e'_{j', i}$ are not null, and $s'_{j, i}\le_{v_iv_{i+1}}e'_{j', i}$.
		\item $P_8(i, i', j)$ returns true if and only if both $s_{i, j}$ and $e_{i', j}$ are not null, and $s_{i, j}\le_{w_jw_{j+1}}e_{i', j}$.
	\end{itemize}
}

\begin{table}[h!]
	\centerline{\begin{tabular}{lcl}
			\hline
			\addlinespace[1pt]
			$P_1(i,j) = \text{true}$ & $\iff$ &  $d(w_{j}, v_{i}v_{i+1})\leq\delta$ \\[.1em] \hline
			\addlinespace[1pt]
			$P_2(i,j) = \text{true}$ & $\iff$ & $d(v_{i},w_{j}w_{j+1}) \leq \delta$\\[.1em] \hline
			\addlinespace[1pt]
			$P_3(i,j,j') = \text{true}$ & $\iff$ &  $\text{neither $s'_{j, i}$ nor $s'_{j', i}$ is null} \, \wedge \, s'_{j, i}\le_{v_iv_{i+1}}s'_{j', i}$\\[.1em] \hline
			\addlinespace[1pt]
			$P_4(i,i',j) = \text{true}$ & $\iff$ & $\text{neither $s_{i, j}$ nor $s_{i', j}$ is null} \, \wedge \, s_{i, j}\le_{w_jw_{j+1}}s_{i', j}$ \\[.1em] \hline
			\addlinespace[1pt]
			$P_5(i, j, j') = \text{true}$ & $\iff$ & $\text{neither $e'_{j, i}$ nor $e'_{j', i}$ is null} \, \wedge \, e'_{j, i}\le_{v_iv_{i+1}}e'_{j', i}$\\[.1em] \hline
			\addlinespace[1pt]
			$P_6(i, i', j) = \text{true}$ & $\iff$ & $\text{neither $e_{i, j}$ nor $e_{i', j}$ is null} \, \wedge \, e_{i, j}\le_{w_jw_{j+1}}e_{i', j}$\\[.1em] \hline
			\addlinespace[1pt]
			$P_7(i, j, j') = \text{true}$ & $\iff$ & $\text{neither $s'_{j, i}$ nor $e'_{j', i}$ is null} \, \wedge \, s'_{j, i}\le_{v_iv_{i+1}}e'_{j', i}$\\[.1em] \hline
			\addlinespace[1pt]
			$P_8(i, i', j) = \text{true}$ & $\iff$ & $\text{neither $s_{i, j}$ nor $e_{i', j}$ is null} \, \wedge \, s_{i, j}\le_{w_jw_{j+1}}e_{i', j}$ \\[.1em]
			\hline
	\end{tabular}}
	\caption{Predicates for determining the signature of a box.}
\end{table}

We describe how to implement them by polynomials.  We view $\delta$ and the coordinates of $v_i, v_{i+1}, v_{i'}$ as fixed coefficienets and the coordinates of $w_j, w_{j+1}, w_{j'}$ as variables.

\vspace{20pt}

\noindent \pmb{$P_1(i,j)$.}
We present several polynomials to implement $P_1(i,j)$.  We use $\langle \nu, \nu' \rangle$ to  denote the inner product of vectors $\nu$ and $\nu'$. Let $\text{aff}(v_{i}v_{i+1})$ be the oriented line that contains $v_{i}v_{i+1}$ and has the same orientation as $v_iv_{i+1}$. We first write $d(w_{j}, \text{aff}(v_{i}v_{i+1}))^2 \cdot \lVert v_{i}-v_{i+1}\rVert^{2}$ as a polynomial. 
\begin{eqnarray}
	& & d(w_{j}, \text{aff}(v_{i}v_{i+1}))^2 \cdot \lVert v_{i}-v_{i+1}\rVert^{2} \nonumber \\
	& = & \lVert w_{j}- v_{i+1}\rVert^2 \cdot \lVert v_{i}-v_{i+1}\rVert^{2}-\langle w_{j}-v_{i+1},v_{i}-v_{i+1} \rangle^2. \label{eq:expand}
\end{eqnarray}

The first polynomial $f_{1,1}^{i,j}$ compares the distance $d(w_{j},\text{aff}(v_{i}v_{i+1}))$ with $\delta$.  
\begin{align*}
	f_{1,1}^{i,j} & = \bigl(d(w_{j}, \text{aff}(v_{i}v_{i+1}))^2-\delta^2\bigr) \cdot \lVert v_{i}-v_{i+1}\rVert^2 \\
	& = \lVert w_{j}- v_{i+1}\rVert^2 \cdot \lVert v_{i}-v_{i+1}\rVert^{2}-\langle w_{j}-v_{i+1},v_{i}-v_{i+1} \rangle^2 - \delta^2 \cdot \norm{v_{i}-v_{i+1}}^2.
\end{align*}
Since $\lVert v_{i}-v_{i+1}\rVert^2$ is positive, we have $d(w_{j},\text{aff}(v_{i}v_{i+1})) \leq \delta$ if and only if $f_{1,1}^{i,j} \leq 0$.  Therefore, if $f_{1,1}^{i,j} > 0$, then  $P_1(i',j')$ is false.  If $f_{1,1}^{i,j} \leq 0$, we use the following two polynomials to check whether the projection of $w_{j}$ in $\text{aff}(v_{i}v_{i+1})$ lies on $v_{i}v_{i+1}$. 
\[
f_{1,2}^{i,j}= \langle w_{j}-v_{i}, v_{i+1}-v_{i} \rangle, \quad\quad f_{1,3}^{i,j}=\langle w_{j}-v_{i+1}, v_{i}-v_{i+1} \rangle.
\]
The projection of $w_{j}$ in $\text{aff}(v_{i}v_{i+1})$ lies on $v_{i}v_{i+1}$ if and only if $f_{1,2}^{i,j} \geq 0$ and $f_{1,3}^{i,j} \geq 0$.  Hence, if $f_{1,1}^{i,j} \leq 0$, $f_{1,2}^{i,j} \geq 0$, and $f_{1,3}^{i,j} \geq 0$, then $P_1(i,j)$ is true.  When $f_{1,2}^{i,j} < 0$ or $f_{1,3}^{i,j} < 0$, we use the following polynomials to compare $d(w_{j},v_{i})$ and $d(w_{j},v_{i+1})$ with $\delta$.
\[
f_{1,4}^{i,j}=\lVert w_{j}-v_{i}\rVert^2-\delta^2, \quad\quad f_{1,5}^{i,j}=\lVert w_{j}-v_{i+1}\rVert^2-\delta^2.
\]
If $f_{1,1}^{i,j} \leq 0$ and $f_{1,2}^{i,j} < 0$, then $v_{i}$ is the nearest point in $v_{i}v_{i+1}$ to $w_{j}$. So, $P_1(i,j)$ is true if and only if $f_{1,4}^{i,j} \leq 0$.  Similarly, in the case that $f_{1,1}^{i,j} \leq 0$ and $f_{1,3}^{i,j} < 0$, $P_1(i,j)$ is true if and only if $f_{1,5}^{i,j} \leq 0$.

\vspace{20pt}   

\noindent \pmb{$P_2(i,j)$.}  We can define polynomials $f_{2,1}^{i,j}$, $f_{2,2}^{i,j}$, $f_{2,3}^{i,j}$, $f_{2,4}^{i,j}$, and $f_{2,5}^{i,j}$ to implement $P_2(i,j)$ in a way analogous to the definitions of polynomials for $P_1(i,j)$. 

\vspace{8pt}

\noindent \pmb{$P_3(i,j,j')$.}  We first check whether $P_1(i, j)$ and $P_1(i, j')$ are true according to the signs of polynomials designed for $P_1(i, j)$ and $P_1(i, j')$. If either $P_1(i, j)$ or $P_1(i,j')$ is false, then either $s'_{j, i}$ or $s'_{j', i}$ is null. Hence, $P_3(i,j, j')$ is false.   Suppose that both $P_1(i, j)$ and $P_1(i, j')$ are true. Let $\tilde{s}'_{j, i}$ be the start of the oriented segment $\+B_{w_j}\cap \mathrm{aff}(v_{i}v_{i+1})$.  Let $\tilde{s}'_{j', i}$ be the start of the oriented segment $\+B_{w_{j'}}\cap \mathrm{aff}(v_{i}v_{i+1})$. We use several polynomials to check whether $\tilde{s}'_{j, i}$ is not behind $\tilde{s}'_{j', i}$ in the direction of $v_{i+1}-v_i$.

We use the polynomial $f_{3,1}^{i,j,j'}$ below to check whether the order of the projections of $w_{j}$ and $w_{j'}$ in $\mathrm{aff}(v_{i}v_{i+1})$ is consistent with the direction of $v_{i+1}-v_i$.
\[
f_{3,1}^{i,j,j'} = \langle w_{j'}-w_{j}, v_{i+1}-v_{i} \rangle.
\]
Let $z_1$ be the distance between the projection of $w_{j}$ in $\text{aff}(v_{i}v_{i+1})$ and the projection of $w_{j'}$ in $\text{aff}(v_{i}v_{i+1})$. Let $z_2$ be the distance between $\tilde{s}'_{j,i}$ and the projection of $w_{j}$ in $\text{aff}(v_{i}v_{i+1})$.  Let $z_3$ be the distance between $\tilde{s}'_{j', i}$ and the projection of $w_{j'}$ in $\text{aff}(v_{i}v_{i+1})$.

If $f_{3,1}^{i, j, j'}\geq 0$, the projection of $w_j$ is not behind the projection of $w_{j'}$. In this case, $\tilde{s}'_{j, i}$ is not behind  $\tilde{s}'_{j', i}$ if and only if $z_1+z_2\ge z_3$. See Figure.~\ref{fig:case1} for illustration. We write $z_1^2$, $z_2^2$ and $z_3^2$ as follows:
\begin{align*}
	z_1^2 &= \langle w_{j'}-w_{j}, v_{i+1}-v_{i} \rangle^2 \cdot \norm{v_{i+1}-v_{i}}^{-2}, \\
	z_2^2 &=\delta^2 - d(w_{j}, \text{aff}(v_{i}v_{i+1}))^2,\\
	z_3^2 &= \delta^2-d(w_{j'}, \text{aff}(v_{i}v_{i+1}))^2.
\end{align*}

\begin{figure}
	\centering
	\includegraphics[scale=0.7]{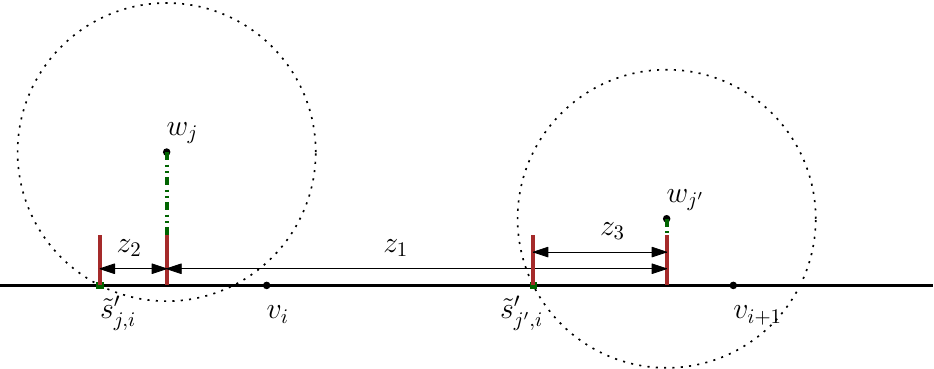}
	\caption{Illustration for the case that the projection of $w_j$ is not behind that of $w_{j'}.$}\label{fig:case1}
\end{figure}

We present several polynomials to check whether $z_1+z_2\geq z_3$. Since $z_1$, $z_2$ and $z_3$ are non-negative, $z_1+z_2\geq z_3$ if and only if $z_1^2+z_2^2+2z_1z_2\geq z_3^2$. By moving $z_1^2+z_2^2$ to the right-hand side, we get $2z_1z_2\geq z_3^2-z_1^2-z_2^2$. We use a polynomial $f_{3,2}^{i, j, j'}$ to represent $(z_3^2-z_1^2-z_2^2)\norm{v_{i+1}-v_{i}}^2$.
\[
f_{3,2}^{i, j, j'}=\bigl( d(w_{j}, \text{aff}(v_{i}v_{i+1}))^2-d(w_{j'}, \text{aff}(v_{i}v_{i+1}))^2 \bigr)\norm{v_{i+1}-v_{i}}^2-\langle w_{j'}-w_{j}, v_{i+1}-v_{i} \rangle^2.
\]
The terms $d(w_{j},\text{aff}(v_{i}v_{i+1}))^2 \norm{v_{i+1}-v_{i}}^2$ and $d(w_{j'}, \text{aff}(v_{i}v_{i+1}))^2 \norm{v_{i+1}-v_{i}}^2$ are expanded as in~\eqref{eq:expand}.  As a shorthand, we do not explicitly show the expansion.

Since $\norm{v_{i+1}-v_{i}}>0$, we have $z_3^3-z_1^2-z_2^2\leq 0$ if and only if $f_{3,2}^{i, j, j'}\leq 0$. Since $2z_1z_2\geq 0$, if $f_{3,2}^{i, j, j'}\leq 0$, we always have $2z_1z_2\ge z_3^2-z_1^2-z_2^2$. It means that $\tilde{s}'_{j, i}$ is not behind $\tilde{s}'_{j', i}$.

Suppose that $f_{3,2}^{i, j, j'}>0$. The inequality $z_1+z_2\geq z_3$ holds if and only if $4z_1^2z_2^2\ge (z_3^2-z_1^2-z_2^2)^2$, i.e. $4z_1^2z_2^2- (z_3^2-z_1^2-z_2^2)^2\geq 0$. We use a polynomial $f_{3,3}^{i, j, j'}$ to denote $(4z_1^2z_2^2- (z_3^2-z_1^2-z_2^2)^2)\norm{v_{i+1}-v_{i}}^4$.
\[
f_{3,3}^{i, j, j'}= 4\langle w_{j'}-w_{j}, v_{i+1}-v_{i} \rangle^2\bigl(\delta^2-d(w_{j},\text{aff}(v_{i}v_{i+1}))^2 \bigr)\norm{v_{i+1}-v_{i}}^2-(f_{3,2}^{i, j, j'})^2
\] 

Since $\norm{v_{i+1}-v_{i}}>0$, $\tilde{s}'_{j, i}$ is not behind $\tilde{s}'_{j', i}$ if and only if $f_{3,3}^{i,j,j'}\geq 0$.

We then consider the case that $f_{3,1}^{i,j,j'} < 0$, i.e., the projection of $w_j$ is behind the projection of $w_{j'}$. In this case, $\tilde{s}'_{j, i}$ is not behind $\tilde{s}'_{j', i}$ if and only if $z_2>z_1+z_3$.  We can define polynomials $f_{3,4}^{i, j, j'}$ and $f_{3,5}^{i, j, j'}$ to check whether $z_2>z_1+z_3$ in the same way as how we check $z_1+z_2\ge z_3$. We finish defining the polynomials to determine whether $\tilde{s}'_{j, i}$ is not behind $\tilde{s}'_{j', i}$ when $s'_{j, i}$ and $s'_{j', i}$ are not null.

We describe how to determine whether $s'_{j, i}=s'_{j', i}$ or $s'_{j, i}\le_{v_iv_{i+1}}s'_{j', i}$. Given that both $s'_{j, i}$ and $s'_{j', i}$ are not null, if $d(v_{i}, w_{j})> \delta$, then $s'_{j, i} = \tilde{s}'_{j, i}$; otherwise, $s'_{j, i} = v_{i}$. Similarly, if $d(v_i,w_{j'}) > \delta$, then $s'_{j',i} = \tilde{s}'_{j',i}$; otherwise, $s'_{j',i} = v_i$.  These decisions can be made using the following polynomials.
\[
f_{3,6}^{i, j}=\norm{v_{i}-w_{j}}^2-\delta^2,\quad\quad
f_{3,7}^{i, j'}=\norm{v_{i}-w_{j'}}^2-\delta^2.
\]
Using these decisions, the comparison between $s'_{j, i}$ and $s'_{j', i}$ falls into the following four cases. If $s'_{j, i}=\tilde{s}'_{j, i}$ and $s'_{j', i}=\tilde{s}'_{j', i}$, we are done.
If $s'_{j, i}=v_{i}$ and $s'_{j', i}=\tilde{s}'_{j', i}$, then $P_3(i, j, j')$ is true because $v_{i}$ is not behind any point on $v_{i}v_{i+1}$ along $v_{i}v_{i+1}$.  If $s'_{j, i} \not= v_i$ and $s'_{j', i}=v_{i}$, then $s'_{j,i} = \tilde{s}'_{j, i}$ is an interior point of $v_iv_{i+1}$, which makes $P_3(i, j, j')$ false.  In the remaining case, both $s'_{j, i}$ and $s'_{j', i}$ are $v_{i}$, so $P_3(i, j, j')$ is true.

\vspace{20pt}

\noindent \pmb{$P_4(i,i',j), P_5(i,j,j'), P_6(i,i',j), P_7(i,j,j'), P_8(i,i',j)$.}  We can define polynomials to implement these predicate in a way analogous to the definitions of the polynomials for $P_3(i,j,j')$.

\section{Predicates and polynomials for Lemma~\ref{lem:coder}}
\label{app:coder}

We want to process $v_{a_k},\ldots,v_{a_{k+1}}$ into a data structure that takes a query point $y$ on a query edge $w_jw_{j+1}$ and locate the immediate predecessor of $y$ in the sorted order of $s_{a_k,j}, e_{a_k,j}, \ldots, s_{a_{k+1},j}, e_{a_{k+1},j}$ along $w_jw_{j+1}$.
The truth values of the following predicates determine the query result.

\begin{table}[h!]
	\centerline{\begin{tabular}{lcl}
			\hline
			\addlinespace[1pt]
			$P_9(i,j) = \text{true}$ & $\iff$ &  $d(v_i, w_{j}w_{j+1})\leq\delta$ \\[.1em] \hline
			\addlinespace[1pt]
			$P_{10}(i,j) = \text{true}$ & $\iff$ & $\text{$s_{i,j}$ is not null}\, \wedge \,y\le_{w_jw_{j+1}}s_{i,j}$\\[.1em] \hline
			\addlinespace[1pt]
			$P_{11}(i,j) = \text{true}$ & $\iff$ &  $\text{$e_{i, j}$ is not null} \, \wedge \, y \le_{w_jw_{j+1}}e_{i,j}$\\[.1em] 
			\hline
			\addlinespace[1pt]
			$P_{12}(i,j) = \text{true}$ & $\iff$ & $s_{i,j}$ is not null and $y = s_{i,j}$ \\[.1em]
			\hline
			\addlinespace[1pt]
			$P_{13}(i,j) = \text{true}$ & $\iff$ & $e_{i,j}$ is not null and $y = e_{i,j}$ \\[.1em]
			\hline
	\end{tabular}}
	\caption{There are $O(\alpha)$ predicates over all $i\in B_k$.}
\end{table}

\cancel{
	\begin{itemize}
		\item $P_9(i, j)$ returns true if and only if $d(v_{i}, w_{j}w_{j+1})\le \delta$.
		\item $P_{10}(i, j)$ returns true if and only if $s_{i, j}$ is not null and $y$ is not behind $s_{i,j}$ along $w_{j}w_{j+1}$.
		\item $P_{11}(i, j)$ returns true if and only if $e_{i, j}$ is not null and $y$ is not behind $e_{i,j}$ along $w_{j}w_{j+1}$.
	\end{itemize}
}


When we invoke a $B_k$-coder and perform the query described above with $y = \ell_{a_k,j}$ during dynamic programming, we already have the signature of the box that involves $B_k$ and $w_jw_{j+1}$.  The signature $\phi$ of the column $j$ of this box tells us the rank values of $s_{a_k,j}, e_{a_k,j}, \ldots, s_{a_{k+1},j}, e_{a_{k+1},j}$ along $w_jw_{j+1}$.  Therefore, the truth values of $P_9,P_{10},P_{11}$ tell us the $s_{i,j}$ with the lowest rank value that lies in front of $y$, which makes $s_{i,j}$ the immediate predecessor.  The truth values of $P_{12}$ and $P_{13}$ tell us whether $y$ is equal to this immediate predecessor.

Next, we construct a set of polynomials for these predicates. We will handle $P_{10}(i,j)$ and $P_{12}(i,j)$ together and $P_{11}(i,j)$ and $P_{13}(i,j)$ together.

\vspace{20pt}

\noindent \pmb{$P_9(i,j)$.} The predicate $P_9(i,j)$ is equivalent to $P_2(i,j)$.  We can use the polynomials $f_{2,1}^{i,j}$, $f_{2,2}^{i,j}$, $f_{2,3}^{i,j}$, $f_{2,4}^{i,j}$, and $f_{2,5}^{i,j}$ to implement it.

\vspace{20pt}

\noindent\pmb{$P_{10}(i,j)$ and $P_{12}(i,j)$.} We first check whether $P_9(i,j)$ is true according to the signs of polynomials designed for $P_9(i,j)$. If $P_9(i,j)$ is false, then $s_{i,j}$ is null. Hence, $P_{10}(i,j)$ and $P_{12}(i,j)$ are false. Suppose that $P_9(i,j)$ is true. In this case, $B_{v_i}\cap w_jw_{j+1}\not=\emptyset$. We use the following polynomial to check whether $y$ is behind the projection of $v_i$ in $\mathrm{aff}(w_jw_{j+1})$ in the direction of $w_{j+1}-w_j$.
\[
f_{10,1}^{i,j} = \langle v_{i}-y, w_{j+1}-w_{j} \rangle.
\]

\begin{figure}
	\centering
	\includegraphics[scale=0.6]{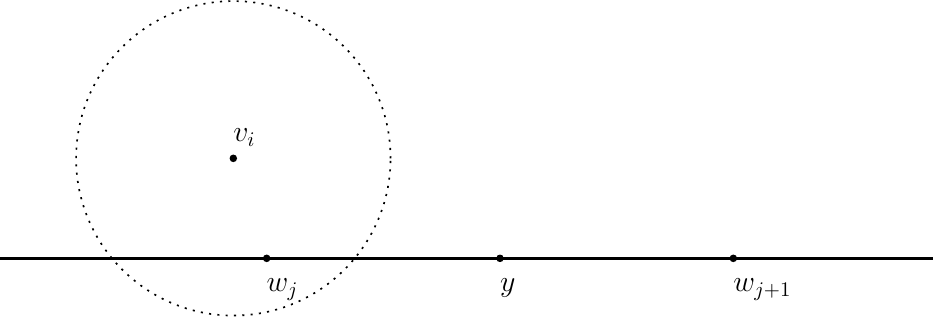}
	\caption{Illustration for the case that the projection of $v_i$ is not in $w_jw_{j+1}$ and $y$ is behind the projection of $v_i$.}
	\label{fig:p10-case1}
\end{figure}

If $f_{10,1}^{i,j}< 0$, $y$ is behind the projection of $v_i$ in $\mathrm{aff}(w_jw_{j+1})$. We use the following two polynomials to check whether the projection of $v_i$ in $\mathrm{aff}(w_jw_{j+1})$ lies on $w_jw_{j+1}$.
\[
f_{10,2}^{i,j}= \langle v_i-w_j, w_{j+1}-w_{j} \rangle, \quad\quad f_{10,3}^{i,j}=\langle v_i-w_{j+1}, w_j-w_{j+1} \rangle.
\]
When $f_{10,2}^{i,j}\ge 0$ and $f_{10,3}^{i,j}\ge 0$, the projection of $v_i$ is in $w_jw_{j+1}$ as well. In this case, $s_{i,j}$ must not be behind the projection of $v_i$. As a result, $y$ is behind $s_{i,j}$, so both $P_{10}(i,j)$ and $P_{12}(i,j)$ are false. In the other case, the projection of $v_i$ is not in $w_jw_{j+1}$. Since $y\in w_jw_{j+1}$ and $y$ is behind the projection of $v_i$, the nearest point in $w_jw_{j+1}$ to $v_i$ is $w_j$. Since $B_{v_i}\cap w_jw_{j+1}\not=\emptyset$, $s_{i,j}$ is equal to $w_j$. See Figure.~\ref{fig:p10-case1} for an illustration. As a result, we need to check whether $y$ is behind $w_j$ in the direction of $w_{j+1}-w_j$ 
by the following polynomial:
\[
f_{10,4}^{i,j}=\langle y-w_j, w_{j+1}-w_j \rangle.
\]
If $f_{10,4}^{i,j}>0$, then $y$ is behind $s_{i,j}$, so both $P_{10}(i,j)$ and $P_{12}(i,j)$ are false. If $f_{10,4}^{i,j}= 0$, then $y$ is equal to $s_{i,j}$, so both $P_{10}(i,j)$ and $P_{12}(i,j)$ are true.

If $f_{10, 1}^i\ge 0$, $y$ is not behind the projection of $v_i$ in $\mathrm{aff}(w_jw_{j+1})$ in the direction of $w_{j+1}-w_j$. Since $y\in w_jw_{j+1}$ and $B_{v_i}\cap w_jw_{j+1}\not=\emptyset$, it must hold that $B_{v_i}\cap yw_{j+1}\not=\emptyset$. In this case, $y$ is not behind $s_{i,j}$ if and only if $d(v_i, y)\ge \delta$. We use the following polynomial to compare $d(v_i,y)$ and $\delta$.
\[
f_{10,5}^{i,j}=\norm{v_i-y}^2-\delta^2.
\]
If $f_{10,5}^{i,j}\ge 0$, the distance between $v_i$ and $y$ is at least $\delta$, then $P_{10}(i,j)$ is true; otherwise, $P_{10}(i,j)$ is false.  If $f_{10,5}^{i,j} = 0$, then $P_{12}(i,j)$ is true; otherwise, $P_{12}(i,j)$ is false.

\vspace{20pt}

\noindent\pmb{$P_{11}(i,j)$ and $P_{13}(i,j)$.} We can define polynomials to implement $P_{11}(i,j)$ and $P_{13}(i,j)$ in a way analogous to the definitions of polynomials for $P_{10}(i,j)$ and $P_{12}(i,j)$.

\bibliography{ref.bib}

\begin{thebibliography}{10}

\bibitem{afshani2018complexity}
Peyman Afshani and Anne Driemel.
\newblock On the complexity of range searching among curves.
\newblock In {\em Proceedings of the ACM-SIAM Symposium on Discrete
  Algorithms}, pages 898--917, 2018.

\bibitem{agarwal2014computing}
Pankaj~K Agarwal, Rinat~Ben Avraham, Haim Kaplan, and Micha Sharir.
\newblock Computing the discrete {F}r{\'e}chet distance in subquadratic time.
\newblock {\em SIAM Journal on Computing}, 43(2):429--449, 2014.

\bibitem{Alt2009}
Helmut Alt.
\newblock {\em The Computational Geometry of Comparing Shapes}, pages 235--248.
\newblock Springer Berlin Heidelberg, Berlin, Heidelberg, 2009.

\bibitem{AG1995}
Helmut Alt and Michael Godau.
\newblock Computing the {F}r\'{e}chet distance between two polygonal curves.
\newblock {\em International Journal of Computational Geometry and
  Applications}, 5:75--91, 1995.

\bibitem{AKW2004}
Helmut Alt, Christian Knauer, and Carola Wenk.
\newblock Comparsion of distance measures for planar curves.
\newblock {\em Algorithmica}, 38:45--58, 2004.

\bibitem{aronov2006frechet}
Boris Aronov, Sariel Har-Peled, Christian Knauer, Yusu Wang, and Carola Wenk.
\newblock Fr{\'e}chet distance for curves, revisited.
\newblock In {\em Proceedings of the 14th Annual European Symposium on
  Algorithms}, pages 52--63. Springer, 2006.

\bibitem{Basu1995OnCA}
Saugata Basu, Richard Pollack, and Marie-Fran\c{c}oise Roy.
\newblock On computing a set of points meeting every cell defined by a family
  of polynomials on a variety.
\newblock {\em Journal of Complexity}, 13:28--37, 1995.

\bibitem{blank2024faster}
Lotte Blank and Anne Driemel.
\newblock A faster algorithm for the {F}r{\'e}chet distance in 1d for the
  imbalanced case.
\newblock {\em arXiv preprint arXiv:2404.18738}, 2024.

\bibitem{bringmann2014walking}
Karl Bringmann.
\newblock Why walking the dog takes time: {F}r\'echet distance has no strongly
  subquadratic algorithms unless {SETH} fails.
\newblock In {\em Proceedings of the IEEE 55th Annual Symposium on Foundations
  of Computer Science}, pages 661--670. IEEE, 2014.

\bibitem{bringmann2015improved}
Karl Bringmann and Marvin K{\"u}nnemann.
\newblock Improved approximation for {F}r{\'e}chet distance on c-packed curves
  matching conditional lower bounds.
\newblock In {\em Proceedings of International Symposium on Algorithms and
  Computation}, pages 517--528. Springer, 2015.

\bibitem{BW2015}
Karl Bringmann and Wolfgang Mulzer.
\newblock Approximiability of the discrete {F}r{\'e}chet distance.
\newblock In {\em Proceedings of the International Symposium on Computational
  Geometry}, pages 739--753, 2015.

\bibitem{bruning2023simplifiedimprovedboundsvcdimension}
Frederik Brüning and Anne Driemel.
\newblock Simplified and improved bounds on the {VC}-dimension for elastic
  distance measures.
\newblock {\em arXiv}, 2023.

\bibitem{buchin2017clustering}
Kevin Buchin, Maike Buchin, David Duran, Brittany~Terese Fasy, Roel Jacobs,
  Vera Sacristan, Rodrigo~I Silveira, Frank Staals, and Carola Wenk.
\newblock Clustering trajectories for map construction.
\newblock In {\em Proceedings of the 25th ACM SIGSPATIAL International
  Conference on Advances in Geographic Information Systems}, pages 1--10, 2017.

\bibitem{buchin2011detecting}
Kevin Buchin, Maike Buchin, Joachim Gudmundsson, Maarten L{\"o}ffler, and Jun
  Luo.
\newblock Detecting commuting patterns by clustering subtrajectories.
\newblock {\em International Journal of Computational Geometry \&
  Applications}, 21(03):253--282, 2011.

\bibitem{buchin2014four}
Kevin Buchin, Maike Buchin, Wouter Meulemans, and Wolfgang Mulzer.
\newblock Four soviets walk the dog—with an application to {A}lt's
  conjecture.
\newblock In {\em Proceedings of the Annual ACM-SIAM Symposium on Discrete
  algorithms}, pages 1399--1413. SIAM, 2014.

\bibitem{buchin2019seth}
Kevin Buchin, Tim Ophelders, and Bettina Speckmann.
\newblock {SETH} says: Weak {F}r{\'e}chet distance is faster, but only if it is
  continuous and in one dimension.
\newblock In {\em Proceedings of the Annual ACM-SIAM Symposium on Discrete
  Algorithms}, pages 2887--2901. SIAM, 2019.

\bibitem{chan2018}
Timothy~M Chan.
\newblock More logarithmic-factor speedups for 3{SUM},(median,+)-convolution,
  and some geometric 3{SUM}-hard problems.
\newblock In {\em Proceedings of the ACM-SIAM Symposium on Discreete
  Algorithms}, pages 881--897, 2018.

\bibitem{chan2018improved}
Timothy~M Chan and Zahed Rahmati.
\newblock An improved approximation algorithm for the discrete {F}r{\'e}chet
  distance.
\newblock {\em Information Processing Letters}, 138:72--74, 2018.

\bibitem{cheng2023solving}
Siu-Wing Cheng and Haoqiang Huang.
\newblock Solving {F}r\'echet distance problems by algebraic geometric methods.
\newblock In {\em Proceedings of the ACM-SIAM Symposium on Discrete
  Algorithms}, pages 4502--4513. SIAM, 2023.

\bibitem{colombe2021approximating}
Connor Colombe and Kyle Fox.
\newblock Approximating the (continuous) {F}r{\'e}chet distance.
\newblock In {\em Proceedings of International Symposium on Computational
  Geometry}. Schloss Dagstuhl-Leibniz-Zentrum f{\"u}r Informatik, 2021.

\bibitem{driemel2012approximating}
Anne Driemel, Sariel Har-Peled, and Carola Wenk.
\newblock Approximating the {F}r{\'e}chet distance for realistic curves in near
  linear time.
\newblock {\em Discrete \& Computational Geometry}, 48:94--127, 2012.

\bibitem{driemel2021vc}
Anne Driemel, Andr\'{e}' Nusser, Jeff~M. Phillips, and Ioannis Psarros.
\newblock The {VC} dimension of metric balls under {F}r{\'e}chet and
  {H}ausdorff distances.
\newblock {\em Discrete \& Computational Geometry}, 66:1351--1381, 2021.

\bibitem{eiter1994computing}
Thomas Eiter and Heikki Mannila.
\newblock Computing the discrete {F}r{\'e}chet distance.
\newblock Technical report, Vienna University of Technology, 1994.

\bibitem{ezra2020decomposing}
Esther Ezra, Sariel Har-Peled, Haim Kaplan, and Micha Sharir.
\newblock Decomposing arrangements of hyperplanes: {VC}-dimension,
  combinatorial dimension, and point location.
\newblock {\em Discrete \& Computational Geometry}, 64(1):109--173, 2020.

\bibitem{F2017}
Ari Freund.
\newblock Improved subquadratic 3{SUM}.
\newblock {\em Algorithmica}, 77(2):440--456, 2017.

\bibitem{Godau1991ANM}
Michael Godau.
\newblock A natural metric for curves - computing the distance for polygonal
  chains and approximation algorithms.
\newblock In {\em Proceedings of the Symposium on Theoretical Aspects of
  Computer Science}, 1991.

\bibitem{GS2017}
Omer Gold and Micha Sharir.
\newblock Improved bounds for 3{SUM}, $k$-{SUM}, and linear degenerarcy.
\newblock In {\em Proceedings of the European Symposium on Algorithms}, pages
  42:1--13, 2017.

\bibitem{gold2018dynamic}
Omer Gold and Micha Sharir.
\newblock Dynamic time warping and geometric edit distance: Breaking the
  quadratic barrier.
\newblock {\em ACM Transactions on Algorithms}, 14(4):1--17, 2018.

\bibitem{GP2018}
Allan Grønlund and Seth Pettie.
\newblock Threesomes, degenerates, and love triangles.
\newblock {\em Journal of the ACM}, 65(4):22:1--22:25, 2018.

\bibitem{gudmundsson2018fast}
Joachim Gudmundsson, Majid Mirzanezhad, Ali Mohades, and Carola Wenk.
\newblock Fast {F}r{\'e}chet distance between curves with long edges.
\newblock In {\em Proceedings of the 3rd International Workshop on Interactive
  and Spatial Computing}, pages 52--58, 2018.

\bibitem{KLM2018}
Daniel~M. Kane, Shachar Lovett, and Shay Moran.
\newblock Near-optimal linear decision trees for $k$-{SUM} and related
  problems.
\newblock In {\em Proceedings of the ACM Symposium on Theory of Computing},
  pages 554--563, 2018.

\bibitem{pollack1993number}
Richard Pollack and Marie-Fran\c{c}oise Roy.
\newblock On the number of cells defined by a set of polynomials.
\newblock {\em Comptes rendus de l'Acad{\'e}mie des sciences. S{\'e}rie 1,
  Math{\'e}matique}, 316(6):573--577, 1993.

\bibitem{vanderhorst_et_al:LIPIcs.SoCG.2024.63}
Thijs van~der Horst and Tim Ophelders.
\newblock {Faster Fr\'{e}chet Distance Approximation Through Truncated
  Smoothing}.
\newblock In {\em Proceedings of the 40th International Symposium on
  Computational Geometry}, pages 63:1--63:15, 2024.

\bibitem{van2023subquadratic}
Thijs van~der Horst, Marc van Kreveld, Tim Ophelders, and Bettina Speckmann.
\newblock A subquadratic $n^\varepsilon$-approximation for the continuous
  {F}r{\'e}chet distance.
\newblock In {\em Proceedings of the Annual ACM-SIAM Symposium on Discrete
  Algorithms}, pages 1759--1776. SIAM, 2023.

\end{thebibliography}
\bibliographystyle{plain}
\end{document}